\documentclass[11pt]{article}

\usepackage{amssymb}
\usepackage{amsmath}
\usepackage{amsthm}
\usepackage{bbm}
\usepackage{bm}
\usepackage{color}
\usepackage{tikz}

\usepackage[margin=1in]{geometry}

\usepackage[backend=bibtex]{biblatex}
\newtheorem{theorem}{Theorem}[section]
\newtheorem{lemma}[theorem]{Lemma}
\newtheorem{definition}[theorem]{Definition}
\newtheorem{corollary}[theorem]{Corollary}

\newtheorem{remark}[theorem]{Remark}

\newcommand\numberthis{\addtocounter{equation}{1}\tag{\theequation}}

\newcommand{\bd}[1]{\mathbf{#1}}
\newcommand{\R}{\mathbb{R}}
\newcommand{\Jn}{\bd{J}^\bd{n}}

\DeclareMathOperator{\vect}{vec}
\DeclareMathOperator{\supp}{supp}
\DeclareMathOperator{\argmin}{argmin}

\newcommand{\vertiii}[1]{{\left\vert\kern-0.25ex\left\vert\kern-0.25ex\left\vert #1 
		\right\vert\kern-0.25ex\right\vert\kern-0.25ex\right\vert}}

\title{Johnson-Lindenstrauss Embeddings with Kronecker Structure}

\newcommand{\footremember}[2]{
	\footnote{#2}
	\newcounter{#1}
	\setcounter{#1}{\value{footnote}}
}
 
\author{
	Stefan Bamberger\footremember{tum}{stefan.bamberger@tum.de, Department of Mathematics, Technical University of Munich}
	\and Felix Krahmer\footremember{tum2}{felix.krahmer@tum.de, Department of Mathematics, Technical University of Munich}
	\and Rachel Ward\footremember{uta}{rward@math.utexas.edu, Department of Mathematics, University of Texas at Austin}
}

\begin{document}

	\maketitle
	
	\begin{abstract}
		We prove the Johnson-Lindenstrauss property for matrices $\Phi D_\xi$ where $\Phi$ has the restricted isometry property and $D_\xi$ is a diagonal matrix containing the entries of a Kronecker product $\xi = \xi^{(1)} \otimes \dots \otimes \xi^{(d)}$ of $d$ independent Rademacher vectors.
		Such embeddings have been proposed in recent works for a number of applications concerning compression of tensor structured data, including the oblivious sketching procedure by Ahle et al.~for approximate tensor computations.
		For preserving the norms of $p$ points simultaneously, our result requires $\Phi$   to have the restricted isometry property for sparsity $C(d) (\log p)^d$. In the case of subsampled Hadamard matrices, this can improve the dependence of the embedding dimension on $p$ to $(\log p)^d$ while the best previously known result required $(\log p)^{d + 1}$. That is, for the case of $d=2$ at the core of the oblivious sketching procedure by Ahle et al., the scaling improves from cubic to quadratic. We provide a counterexample to prove that the scaling established in our result is optimal under mild assumptions.  
	\end{abstract}

	
	\section{Introduction}
	
	Johnson-Lindenstrauss (JL) embeddings, first discovered in \cite{jloriginal} and subsequently re-introduced in distributional form  \cite{FM, linial1995geometry, arriaga2006algorithmic, achlioptas2001database,  Dasgupta}, provide a random embedding of finitely many points into a lower-dimensional vector space while preserving the structure of these points, i.e.~their pairwise Euclidean distances. This property has been proven to be useful for reducing the complexity of algorithms in many fields, such as numerical linear algebra or machine learning.
	
	In particular, the technique of \emph{sketching} --  for which \cite{sketching_book} provides a detailed overview -- uses dimension reduction transforms such as JL embeddings to reduce the complexity of problems in numerical linear algebra. For example, instead of solving the classical linear regression problem $\min_{x} \|A x - b\|_2^2$, one can apply a Johnson-Lindenstrauss embedding $\Phi$ to $b$ and the columns of $A$ which leads to a smaller-dimensional problem $\min_{x} \|\Phi A x - \Phi b\|_2^2$ which can often  be solved more efficiently.  The Johnson-Lindenstrauss assumption is a simple sufficient condition under which the solution of the reduced problem is guaranteed to yield a good approximation to the original problem \cite{sarlos2006improved}. 
	
	In response to the driving application of improving algorithmic complexity of sketched linear algebra  problems at massive scale,  a line of research on \emph{fast} Johnson-Lindenstrauss embeddings emerged, concerning the construction and analysis of random matrices $\Phi$ with the Johnson-Lindenstrauss property and which also have structure allowing for fast matrix-vector multliplication.  This analysis was initiated with the fast JL transform introduced in \cite{fastjl}, in the form of a randomly row-subsampled discrete Hadamard matrix with randomized column signs.  This construction was later improved and refined in  \cite{ jlcirc, ailib}, and ultimately sharpened to the best-known embedding power in \cite{riptojl} by establishing a near-equivalence between the  Johnson-Lindenstrauss embedding property and a deterministic restricted isometry property \cite{riptojl}. Recently, this line of work found new energy following the work \cite{BBK18} which proposed the use of a row-subsampled discrete Hadamard matrix with column signs randomized according to a Kronecker-structured Rademacher vector, and conjectured that such an embedding satisfies the Johnson-Lindenstrauss property.  The Kronecker structure allows for even faster matrix-vector multiplication when applied to data points with Kronecker structure themselves, as arise naturally when dealing with multidimensional data arrays  (see, for example, applications to kernel methods with polynomial sketching \cite{ahle2020oblivious}, and solving least squares problems with tensor structure \cite{kronecker_jl, iwen2019lower}).  
	Indeed, suppose we want to embed a data point $x = x^{(1)} \otimes \dots \otimes x^{(d)} \in \R^{n^d} = \R^N$ which is a Kronecker product of $d$ data vectors, each of dimension $n$.  If the embedding matrix $\Phi \in \R^{m \times N}$ itself has Kronecker structure $\Phi = \Phi^{(1)} \otimes \dots \otimes \Phi^{(d)}$ where the dimensions of the factors of $\Phi$ correspond to the factor dimensions of $x$, then the matrix-vector multiplication $\Phi x$ can be factored as $\Phi x = (\Phi^{(1)} x^{(1)}) \otimes \dots \otimes (\Phi^{(d)} x^{(d)})$, and \emph{can be computed factor by factor, without constructing $x$ explicitly}.  Because the Kronecker product of discrete Hadamard matrices is itself a discrete Hadamard matrix, embedding matrices in the form of discrete Hadamard matrices with Kronecker-structured random column signs fall within this framework, and it is natural to study the embedding power of such transforms.   In this paper, we improve, simplify, and generalize the current embedding results for the Kronecker Johnson-Lindenstrauss embedding \cite{kronecker_jl, ahle2020oblivious, malik2020guarantees, iwen2019lower} by generalizing an approach from \cite{riptojl} on near-equivalence between Johnson-Lindenstrauss property and the restricted isometry property to JL embeddings with Kronecker structure to higher-degree tensor embeddings.

	\subsection{Background and prior work}
	Recall the distributional version of the JL Lemma: for any $\epsilon > 0$ and $\eta < 1/2$ and positive integer $N$, there exists a distribution over $\mathbb{R}^{m \times N}$ such that for a fixed unit-length vector $x \in \mathbb{R}^N$ and for a random matrix $\Phi$ drawn from this distribution with $m = O(\epsilon^{-2} \log(1/\eta))$,
	\begin{equation}
		\label{JL:dist}
		\mathbb{P}( \left| \| \Phi x \|_2^2 - 1 \right| > \epsilon) < \eta.
	\end{equation}
	The dependence $m = O(\epsilon^{-2} \log(1/\eta))$, as achieved by (properly normalized) random matrices with independent and identically distributed subgaussian entries  \cite{Dasgupta},  is tight, as shown recently in \cite{larsen2017optimality} improving on a previous (nearly-tight) lower bound \cite{alon}.  
	
	For a given $\Phi: \mathbb{R}^N \rightarrow \mathbb{R}^m$ generated as such, computing the matrix-vector product $x \rightarrow \Phi x$ has time complexity $O( m N )$.  The fast Johnson-Lindenstrauss as introduced in \cite{fastjl} and improved in \cite{ailib, riptojl}, is constructed by randomly flipping the column signs of a random subset of $m$ rows from the $N \times N$ Discrete Fourier (or Discrete Hadamard) Transform.  Exploiting the FFT algorithm, the fast JLT computes a matrix-vector product in time $O(N \log(N))$. The trade-off for this time savings is that the fast JLT has reduced embedding power $m = O(\epsilon^{-2} \log(1/\eta)\log^3(N)).$
	
	More recently, the Kronecker fast JL transform (KFJLT)  was proposed in \cite{BBK18}, to further improve the algorithmic complexity of the fast JL embedding in applications to Kronecker-structured data. 
	
	Such a construction has found applications as a key ingredient of  the oblivious sketching procedure \cite{ahle2020oblivious}, a multiscale construction for dimension reduction appplicable for subspace embeddings and approximate matrix multiplication. A central idea of this construction is the repeated application of the  Kronecker FJLT of order $d = 2$.

	The KFJLT of order $d$ acts on a Kronecker-structured vector $x = x^{(1)} \otimes \dots \otimes x^{(d)} \in \R^{n_1 \dots n_d} = \R^N$ as follows:  For fixed diagonal matrices with i.i.d. Rademacher random variables  $D_{(1)}, \dots, D_{(d)}$ of dimensions $n_1, \dots, n_d$ respectively, and for a random subset $\Omega \subset [N]$ of size $|\Omega| = m$:
	\begin{enumerate}
		\item Randomly flip signs of the entries in each vector factor according to $x^{(k)} \rightarrow D_{(k)} x^{(k)} =: z^{(k)}$;
		\item Compute the DFTs of each factor $y^{(k)} = H_{n_k} z^{(k)}$, where $H_{n}$ is the $n \times n$ DFT matrix (normalized to be a unitary transform).  
		\item Compress $y = y^{(1)} \otimes \dots \otimes y^{(d)} \in \mathbb{C}^N$ to  $y_{\Omega} \in \mathbb{C}^m$, where $y_{\Omega}$ consists of the entries in $y$ restricted to the subset $S$
		\item Rescale $y_{\Omega}$ by $\sqrt{N/m}$.
	\end{enumerate}
	The Kronecker JL transform extends to a well-defined linear map for any input $x \in \mathbb{R}^N$, taking the form of a matrix which can be expressed as the product of three matrix types:
	\begin{equation}
		\label{kfjlt}
		\sqrt{\frac{N}{m}} \cdot P_{\Omega} \cdot H \cdot D_{\xi} \in \mathbb{R}^{m \times N}
	\end{equation}
	where $D_{\xi}$ is the $N \times N$ diagonal matrix with diagonal vector $\xi = \xi_1 \otimes \xi_2 \otimes \dots \otimes \xi_d$ the Kronecker product of $n_1, \dots, n_d$-dimensional Rademacher vectors,  $H = H_{n_1} \otimes \dots \otimes H_{n_d}$ is the Kronecker product of orthonormal DFTs (or, more generally, of bounded orthogonal matrices, including DFTs, Hadamard, etc .... ), and $P_{\Omega}: \mathbb{R}^N \rightarrow \mathbb{R}^m$ denotes the projection matrix onto the coordinate subset $\Omega$.  Our results hold for more general constructions of Kronecker products of matrix factors satisfying the restricted isometry property, after randomizing their column signs. 
	
	In the special case $d=1$, the Kronecker FJLT reduces to the standard FJLT as considered in \cite{riptojl}.  However, when the input vector has Kronecker structure so that the mapping can be applied separately to the matrix-vector factors, the complexity of computing a KJLT transform matrix-vector product improves to $O( n_1 \log(n_1) + \dots n_d \log(n_d)  + m d)$.  The price that is paid is that the embedding power (that is, the minimal scaling of the embedding dimension in $m$ necessary for the distributional JL \eqref{JL:dist}) is weakened by the reduced randomness in $\xi$. For a numerical demonstration of the suboptimal scaling, we refer the reader to \cite{kronecker_jl}. A general theoretical lower bound was, to our knowledge, not available before this paper; lower bounds for related but somewhat different constructions were shown in \cite{ahle2020oblivious}.
	
	At the same, a number of works have investigated sufficient conditions on the embedding dimension to ensure that the map given by \eqref{kfjlt} satisfies the distributional JL property \eqref{JL:dist}.
	The papers \cite{ahle2020oblivious, malik2020guarantees}
	show that for \eqref{kfjlt} based on the Hadamard transform, a sufficient condition is given by 
	$m = C_d \cdot \frac{1}{\epsilon^2}\log(1/\eta)^{d+1}$, up to logarithmic factors in $\log(1/\eta), 1/\epsilon$, and $N$. While the analysis in \cite{malik2020guarantees} is restricted to vectors with a Kronecker structure, the generalization of \cite{ahle2020oblivious} applies to arbitrary vectors.

	
	On the other hand, the paper \cite{kronecker_jl} used the near-equivalence between JL embedding and restricted isometry property from \cite{riptojl} to provide the sufficient condition $m = C_d \frac{1}{\epsilon^2} (\log(1/\eta))^{2d-1}$ for any subsampled bounded orthonormal transform, thus including but not limited to constructions based on Hadamard transform.  
	
	To put these two results into perspective, we remind the reader that the tensor degree $d$ is typically small -- recall that the oblivious sketching procedure of \cite{ahle2020oblivious} only uses the case $d=2$, where the two conditions basically agree. Hence also for our results, we will pay special attention to optimizing the dependence for small values of $d$.
	
	\subsection{Contributions of this work} 
	
	In this work, we improve the existing bounds on the embedding dimension for the general  Kronecker FJLT  to $m = C_{d} \frac{1}{\epsilon^2} (\log(1/\eta))^{d},$ up to logarithmic factors in $\log(1/\eta),\,1/\epsilon$, and in $N$, improving the results in  \cite{ahle2020oblivious}  by a factor of $\log(1/\eta)$. In particular, for the case of $d=2$ at the core of the oblivious sketching procedure \cite{ahle2020oblivious}, our results improve the scaling of the embedding dimension in $\log(\tfrac{1}{\eta})$ from cubic to quadratic.
	
	We additionally prove that this embedding result is optimal in the $\eta$ dependence by providing a lower bound of  $m = \Theta( (\log(1/\eta))^{d})$  in Section \ref{sec:lower_bounds}. 
	We achieve the optimal bounds by generalizing the near-equivalence between the JL property and the restricted isometry property of \cite{riptojl} to higher-order tensors, in a sharper way than what was shown in \cite{kronecker_jl}, by carefully using a higher-dimensional analog of the Hanson-Wright inequality for random tensors.   
	
	We state our main results in Section~\ref{sec:main_results}. Then we summarize the required existing tools and describe the main idea behind our proof in Section~\ref{sec:prelim_proof_overview} without covering technical details. The full proof including all technical steps is then given in Section~\ref{sec:proof_general}. Then in Section~\ref{sec:lower_bounds}, we show the aforementioned optimality of our result and then conclude by discussing the implications of our work in Section~\ref{sec:discussion}.
	
	\subsection{Related work}
	
	
	
	Tensor Johnson-Lindenstrauss constructions have become a recent topic of study, even beyond the concrete construction of \eqref{kfjlt}.  
	
	Tensor JL embeddings based on sparse matrix structure have been studied in the context of vectors with Kronecker structure, based on the count sketch technique \cite{charikar2002finding}, which has been extended to the tensorized version known as tensor sketch in \cite{pham2013fast}. Applications to problems including subspace embeddings and approximate matrix multiplication are presented in \cite{sparse_tensor_embedding}. However, these methods have a worse dependence on the failure probability compared to Kronecker FJLT.

	The paper \cite{iwen2019lower} derived fast tensor embeddings for subspaces.  
	The paper \cite{sun2021tensor} proposed tensor random projections as matrices whose rows are i.i.d.~Kronecker products of independent Gaussian vectors, and proved embedding properties for such constructions for Kronecker products of order $d=2$.  The paper \cite{ahle2020oblivious} extended the analysis beyond $d = 2$, and  \cite{chen2020nearly} further refined and extended these results in the context of sketching constrained least squares problems.

	
	\subsection{Notation}
	
	$Id_N \in \R^{N \times N}$ is the identity matrix. We denote $H \in \mathbb{R}^{N \times N}$ for the normalized Hadamard matrix where $N$ is a power of $2$. This is obtained recursively as $H_{\log_2 N}$ by $H_0 = (1)$ and $H_{k + 1} = \frac{1}{\sqrt{2}}\begin{pmatrix}H_k & H_k \\ H_k & -H_k \end{pmatrix}$ such that $H$ is orthogonal, i.e., $H^T H = Id_N$.
	
	For an $\R$-valued random variable $X$, we define $\|X\|_{L_p} := (\mathbb{E}|X|^p)^\frac{1}{p}$. We define the subgaussian norm $\|X\|_{\psi_2} := \sup_{p \geq 1} \|X\|_{L_p} / \sqrt{p}$. For a random vector $Y \in \R^N$, we define the subgaussian norm $\|Y\|_{\psi_2} := \sup_{v \in \R^N, \|v\|_2 = 1} \| \langle v, Y \rangle \|_{L_p}$. We call $Y$ isotropic if $\mathbb{E} Y Y^T = Id_N$.
	
	A random vector with independent entries which are $\pm 1$ with probability $\frac{1}{2}$ each, is called a Rademacher vector.
	
	We say that a matrix $\Phi \in \mathbb{R}^{m \times N}$ satisfies the $(s, \delta)$-restricted isometry property (RIP) if $(1 - \delta) \|x\|_2^2 \leq \|\Phi x\|_2^2 \leq (1 + \delta) \|x\|_2^2$ holds for all $s$-sparse $x \in \R^N$, i.e., all $x$ with at most $s$ non-zero entries.
	
	For a subset $S \subset [d]$, we denote $S^c := [d] \backslash S$.
	
	To formulate the main proof, we will need some additional notation for handling higher order arrays which will be introduced in Section \ref{sec:array_indices}.

	\section{Main result} \label{sec:main_results}
	
	\begin{definition} \label{def:dist_jl}
		For $\epsilon, \eta > 0$, a random matrix $A \in \R^{m \times N}$ satisfies the $(\epsilon, \eta)$ distributional Johnson-Lindenstrauss property if for all $x \in \R^N$ with $\|x\|_2 = 1$,
		\[
		\mathbb{P}\left( \left| \|A x\|_2^2 - 1 \right| > \epsilon \|x\|_2^2 \right) \leq \eta.
		\]
	\end{definition}
	
	\begin{remark}
		If $A \in \R^{m \times N}$ has the $(\epsilon, \frac{\tilde{\eta}}{p (p - 1)})$ distributional Johnson-Linden\-strauss property, then for any set $E \subset \R^N$ with $|E| = p$ elements, by a union bound we obtain
		\begin{align*}
			& \mathbb{P}\left( \exists x, y \in E: \left| \|A x - A y\|_2^2 - \|x - y\|_2^2 \right| > \epsilon \|x - y\|_2^2 \right) \\
			& \leq \sum_{\substack{x, y \in E \\ x \neq y}} \mathbb{P}\left( \left| \|A \frac{x - y}{\|x - y\|_2}\|_2^2 - 1 \right| > \epsilon  \right) \leq 
			|E| (|E| - 1) \cdot \frac{\tilde{\eta}}{p (p - 1)} = \tilde{\eta}.
		\end{align*}
		
		So with a probability of at least $1 - \tilde{\eta}$, it holds that
		\[
		\forall x, y \in E: \left| \|A x - A y\|_2^2 - \|x - y\|_2^2 \right| \leq \epsilon \|x - y\|_2^2.
		\]
		Then $A$ preserves all pairwise distances in the set $E$ up to a factor of $1 \pm \epsilon$.
	\end{remark}
	
	\begin{theorem} \label{thm:main}
		For $d \geq 1$, let $n_1, \dots, n_d$ be dimensions such that $N = n_1 \dots n_d$. Let $0 < \epsilon, \eta < 1$ and $\Phi \in \R^{m \times N}$ be a matrix satisfying the $(2 s^d, \delta)$-RIP for $s \geq \log\frac{1}{\eta} + 2$ and $\delta \leq C(d) \epsilon$ where $C(d)$ is a constant that only depends on $d$.
		
		Let $\xi^{(1)} \in \{\pm 1\}^{n_1}, \dots, \xi^{(d)} \in \{\pm 1\}^{n_d}$ be independent Rademacher vectors and $\xi := \xi^{(1)} \otimes \dots \otimes \xi^{(d)} \in \R^N$. Define $A := \Phi D_{\xi} \in \R^{m \times N}$ where $D_\xi$ is a diagonal matrix with the entries of $\xi$ on its diagonal.
		
		Then $A$ satisfies the $(\epsilon, \eta)$ distributional Johnson-Lindenstrauss property.
	\end{theorem}
	
	Using the result by Haviv and Regev about the RIP \cite{fourierrip}, we obtain the following result for such embeddings using subsampled bounded orthonormal matrices as $\Phi$ such as a subsampled Kronecker product of $d$ Hadamard matrices.

	\begin{corollary} \label{cor:jl_hadamard}
		Let $n_1, \dots, n_d, N, \epsilon, \eta, \xi$ be as in Theorem \ref{thm:main}, $\nu \in (0, 1)$, $\Phi = \sqrt\frac{N}{m} P_\Omega F \in \R^{m \times N}$ where $P_\Omega \in \R^{m \times N}$ represents uniform independent subsampling of rows with replacement and $F \in \R^{N \times N}$ is a unitary matrix with entries bounded by $\frac{D}{\sqrt{N}}$ in absolute value. 
		
		If $N \geq \frac{1}{(\nu)^{C_1 d \log \log(\frac{1}{\nu})}}$ and
		\[
		m \geq C(d) D^2 \epsilon^{-2} \left( \log\frac{1}{\eta} \right)^d \left( \log\frac{C(d)}{\epsilon} \right)^2 \log N \left( \log\frac{C(d) \log\frac{1}{\eta}}{ \epsilon} \right)^2,
		\]
		then with probability $\geq 1 - \nu$ (with respect to $P_\Omega$), we obtain a matrix $\Phi$ such that $\Phi D_\xi$ satisfies the $(\epsilon, \eta)$ distributional Johnson-Lindenstrauss property (with respect to the probability in $\xi$).
		
		$C_1$ is an absolute constant and $C(d)$ only depends on $d$.
	\end{corollary}
	
	\begin{remark}
		For norm preservation of $p$ points simultaneously through a union bound, an $(\epsilon, \frac{c}{p})$ distributional Johnson-Lindenstrauss property is required for a constant $c \in (0, 1)$. The RIP is a property that holds uniformly for all sparse vectors such that in Corollary \ref{cor:jl_hadamard}, no union bound over the probability in $P_\Omega$ is required and $\nu$ can be chosen to be constant and especially independent of $p$.
		
		So even though the lower bound on $N$ in Corollary \ref{cor:jl_hadamard} implies $\log N \gtrsim \log\frac{1}{\nu}$ in the formula for the lower bond on $m$, the dependence of $m$ on $p$ will only be $m \gtrsim \left( \log p \right)^{d}$.
	\end{remark}

	\section{Preliminaries and proof overview} \label{sec:prelim_proof_overview}
	
	In this section, we provide the required probabilistic tools and give an overview of the proof.
	
	\cite{riptojl} provides a proof for the $d = 1$ case based on a union bound and the Hanson-Wright inequality. This method is generalized to higher $d$ in \cite{kronecker_jl} by successively applying the Hanson-Wright inequality.
	
	The Hanson-Wright inequality has also been generalized to chaos of higher order in \cite{chaos_prob_ub} for a Gaussian chaos. To state this bound, we introduce the following notation. Let $\bd{B} \in \mathbb{R}^{n_1 \times \dots \times n_d}$ and denote $S(d, \kappa)$ for the set of all partitions of $[d]$ into $\kappa$ nonempty disjoint sets. The entries in $\bd{B}$ are indexed by indices in $\bd{J} := [n_1] \times \dots \times [n_d]$. For an index $\bd{i} \in \bd{J}$ and $S \subset [d]$, define $\bd{i}_S := (\bd{i}_j)_{j \in S}$. For $(I_1, \dots, I_\kappa) \in S(n, \kappa)$, define
	\begin{equation}
		\|\bd{B}\|_{I_1, \dots, I_\kappa} := \sup \left\{ \sum_{\bd{i} \in \bd{J}} {B_{\bd{i}} \alpha_{\bd{i}_{I_1}}^{(1)} \dots \alpha_{\bd{i}_{I_\kappa}}^{(d)} } \,|\, \|\alpha^{(1)}\|_2 = \dots = \|\alpha^{(d)}\|_2 = 1  \right\}.
		\label{eq:partition_norm}
	\end{equation}
	
	The result considers the case in which $n_1 = n_2 = \dots = n_d$. However, this is not an essential restriction since in the case of varying dimensions, $B$ can be extended by $0$ entries to an array with $\max_{k \in [d]} n_k$ dimensions along every axis. Then the following moment bounds are given in Theorem 1 of \cite{latala_gauss_chaos}.
	
	\begin{theorem}[Theorem 1 in \cite{latala_gauss_chaos}] \label{thm:chaos_moments}
		Let $\bd{B} \in \R^{n_1 \times \dots \times n_d}$ and let $g^{(1)} \in \R^{n_1}, \dots, \allowbreak g^{(d)} \in \R^{n_d}$ be independent standard normal vectors. Define
		\[
		X := \sum_{i_1 \in [n_1], \dots, i_d \in [n_d]}{B_{i_1, \dots, i_d} g_{i_1}^{(1)} \dots g_{i_d}^{(d)} }.
		\]
		Then for any $p \geq 2$,
		\[
		\|X\|_{L_p} \leq C(d) \sum_{\kappa = 1}^d p^{\kappa/2} \sum_{(I_1, \dots, I_\kappa) \in S(d, \kappa)} \|\bd{B}\|_{I_1, \dots, I_k}.
		\]
		where $C(d)$ is a constant that only depends on $d$.
	\end{theorem}
	
	Note that in the case $d = 2$, only the two norms $\|\cdot\|_{\{1, 2\}}$ and $\|\cdot\|_{\{1\}, \{2\}}$ are involved which are equal to the Frobenius and spectral norm of a matrix, respectively. Then the bound in Theorem \ref{thm:chaos_moments} corresponds to the one in the classical Hanson-Wright inequality.
	
	In our setting, we need to control analogous expressions with Rademacher variables instead of Gaussian ones. However, the above result can be generalized to subgaussian random vectors, for example by repeated application of Lemma 3.9 in \cite{hanson_wright_tensors} which compares the subgaussian chaos to the Gaussian one. Together with Theorem \ref{thm:chaos_moments} this yields the following statement.
	
	\begin{theorem} \label{thm:subgaussian_decoupled}
		Let $\bd{n} \in \mathbb{N}^d$, $\bd{B} \in \mathbb{R}^{\bd{n}}$, $p \geq 2$.
		
		Let $S(\kappa, d)$ denote the set of partitions of $[d]$ into $\kappa$ nonempty disjoint subsets.  Define
		\begin{align*}
			m_p(\bd{B}) := \sum_{\kappa = 1}^d p^{\kappa / 2} \sum_{(I_1, \dots, I_\kappa) \in S(k, d)} \|\bd{B}\|_{I_1, \dots, I_\kappa}.
		\end{align*}
		
		Consider independent, mean $0$, isotropic vectors $X^{(1)} \in \R^{n_1}, \dots, X^{(d)} \in \R^{n_d}$ with subgaussian norm bounded by $L \geq 1$. Then
		\begin{align*}
			\left\| \sum_{i_1 \in [n_1], \dots, i_d \in [n_d]}{B_{i_1, \dots, i_d} X_{i_1}^{(1)} \dots X_{i_d}^{(d)} } \right\|_{L_p} \leq C(d) L^{d} m_p(\bd{B}),
		\end{align*}
		where $C(d) > 0$ is a constant that only depends on $d$.
	\end{theorem}
	
	Note that in the scenario considered in this work, we will consider such chaos expressions in which there are $d$ pairs of equal vectors which is not exactly the same as in Theorem \ref{thm:subgaussian_decoupled}. However, it is possible to reduce the problem in this work to the setting of Theorem \ref{thm:subgaussian_decoupled} using a decoupling method. Such a decoupling method has been studied in the same work \cite{hanson_wright_tensors} and it will be discussed in Section \ref{sec:decoupling} together with its application in this work.
	
	Similarly to what is done in \cite{latala_gauss_chaos}, also the moment bound in Theorem \ref{thm:subgaussian_decoupled} can be translated into a tail bound. The following auxiliary result states a general relation between moment and tail bounds.
	
	\begin{lemma}[Lemma 3.8 in \cite{hanson_wright_tensors}] \label{lem:moment_tail_bound}
		Let $T$ be a finite set and $X$ an $\R$ valued random variable such that for all $p \geq p_0 \geq 0$,
		\[
		\|X\|_{L_p} \leq \sum_{k = 1}^{d} \min_{l \in T} p^{e_{k, l}} \gamma_{k, l}
		\]
		for values $\gamma_{k, l} > 0$.
		
		Then for all $t > 0$,
		\[
		\mathbb{P}(|X| > t) \leq e^{p_0} \exp\left( - \min_{k \in [d]} \max_{l \in T} \left( \frac{t}{e d \gamma_{k, l}} \right)^\frac{1}{e_{k, l}} \right).
		\]
	\end{lemma}

	We make use of the restricted isometry through the following lemma which is used in the proof of Lemma \ref{lem:inner_pro_bound}. For a more general overview of the restricted isometry property and similar tools, see Chapter 6 in \cite{foucart_rauhut}.
	\begin{lemma} \label{lem:rip_spec_norm}
		Let $\Phi \in \R^{m \times N}$ have the $(2s, \delta)$-RIP. Then for any $S, T \subset [N]$ of size $|S| = |T| = s$, the submatrix $B = (\Phi^* \Phi - Id_N)_{S, T}$ satisfies $\|B\|_{2 \rightarrow 2} \leq \delta$.
	\end{lemma}
	\begin{proof}
		Let $x, y \in \R^N$ such that $\supp(x) = S$, $\supp(y) = T$ and $\|x\|_2 = \|y\|_2 = 1$. Then by the polarization identity and the RIP
		\begin{align*}
			|x_S^* B y_T| & = |x^* \Phi^* \Phi y - x^* y| \\
			& =
			\frac{1}{4} \left| \|\Phi(x + y)\|_2^2 - \|\Phi(x - y)\|_2^2 - \|x + y\|_2^2 + \|x - y\|_2^2 \right| \\
			& \leq \frac{\delta}{4} \left( \|x + y\|_2^2 + \|x - y\|_2^2 \right)
			= \frac{\delta}{4}\left( 2 \|x\|_2^2 + 2 \|y\|_2^2 \right) = \delta.
		\end{align*}
	\end{proof}
	
	\subsection{Proof overview} \label{sec:proof_outline}
	
	In this section we will give an overview and intuition for the proof of the main result, Theorem \ref{thm:main}. The complete technical proof can be found in Section \ref{sec:proof_general}. We will also point out the analogous steps in the order $1$ case as in \cite{riptojl}.
	
	Since $A$ is a linear map, we can assume $\|x\|_2 = 1$ and it is sufficient to prove $\left| \|A x\|_2^2 - 1 \right| \leq \epsilon$ with probability $\geq 1 - \eta$ for each such $x$ to show the distributional Johnson-Lindenstrauss property.
	
	The key to the proof is relating the norm $\|A x\|_2^2$ to a Rademacher chaos. In the case $d = 1$, we obtain 
	\begin{align*}
		\|A x\|_2^2 = \sum_{j = 1}^m \left( \sum_{k = 1}^N \Phi_{j, k} \xi_k x_k  \right)^2 = \sum_{j = 1}^m \sum_{k, l = 1}^N \Phi_{j, k} \Phi_{j, l} \xi_k \xi_l x_k x_l = \sum_{k, l = 1}^N \xi_k \xi_l B_{k, l}
	\end{align*}
	where $B_{k, l} := \sum_{j = 1}^m  \Phi_{j, k} \Phi_{j, l}  x_k x_l$ ($B \in \R^{N \times N}$). This is a Rademacher chaos of order $2$ and can be controlled with the Hanson-Wright inequality which is done in \cite{riptojl}.
	
	In the general case, let $L: [n_1] \times \dots \times [n_d] \rightarrow [N]$ be the function that maps a tuple of indices of $\xi^{(1)}, \dots, \xi^{(d)}$ to the corresponding index in $\xi^{(1)} \otimes \dots \otimes \xi^{(d)}$. Then
	\begin{align}
		& \|A x\|_2^2 = \sum_{j = 1}^m \left( \sum_{\bd{k} \in [n_1] \times \dots \times [n_d]} \Phi_{j, L(\bd{k})} \xi^{(1)}_{\bd{k}_1} \dots \xi^{(d)}_{\bd{k}_d} x_{L(\bd{k})}  \right)^2 \nonumber \\
		= & \sum_{\substack{\bd{k}, \bd{k}' \in [n_1] \times \dots \times [n_d]}} \xi^{(1)}_{\bd{k}_1} \dots \xi^{(d)}_{\bd{k}_d} \xi^{(1)}_{\bd{k}_1'} \dots \xi^{(d)}_{\bd{k}_d'} x_{L(\bd{k})} x_{L(\bd{k}')} \sum_{j = 1}^m \Phi_{j, L(\bd{k})} \Phi_{j, L(\bd{k}')} \nonumber \\
		= & \sum_{\substack{\bd{k}, \bd{k}' \in [n_1] \times \dots \times [n_d]}} \xi^{(1)}_{\bd{k}_1} \dots \xi^{(d)}_{\bd{k}_d} \xi^{(1)}_{\bd{k}_1'} \dots \xi^{(d)}_{\bd{k}_d'} B_{\bd{k}, \bd{k}'} \label{eq:order2d_chaos}
	\end{align}
	for a corresponding array $\bd{B} \in \R^{n_1 \times \dots \times n_d \times n_1 \times \dots \times n_d}$. This is a Rademacher chaos of order $2 d$ and with some adaptions, it can be controlled with Theorem \ref{thm:subgaussian_decoupled}.
	
	The difference between (\ref{eq:order2d_chaos}) and the expression controlled in Theorem \ref{thm:subgaussian_decoupled} is that in the former one, each $\xi^{(j)}$ appears twice, while in the latter one, all the involved random vectors have to be independent. This can be overcome using a decoupling technique. For the order $2$ chaos, the classical decoupling lemma ensures that instead of $\sum_{k, l = 1}^N \xi_k \xi_l B_{k, l}$, controlling the decoupled chaos $\sum_{k, l = 1}^N \xi_k \xi_l' B_{k, l}$ is sufficient where $\xi'$ is an independent copy of $\xi$. This step only deteriorates the final result by a constant factor. An analogous result can be shown for the general order $2 d$ chaos as in (\ref{eq:order2d_chaos}) such that we can replace the second appearance of every Rademacher vector by an independent copy and then Theorem \ref{thm:subgaussian_decoupled} can be applied. The detailed statements of this step will be given in Section \ref{sec:decoupling}.
	
	A crucial step to control $\|A x\|_2^2$ in \cite{riptojl} for $d = 1$ is the separation of the largest entries of $x$ (in absolute value), i.e., $x = x_{(1)} + x_{(\flat)}$ where $x_{(1)}$ contains the largest $s \sim \log(p)$ entries in absolute value and $x_{(\flat)}$ the remaining ones. Then
	\begin{align}
		\|A x\|_2^2 & = \langle A(x_{(1)} + x_{(\flat)}), A(x_{(1)} + x_{(\flat)}) \rangle \nonumber \\
		& = \|A x_{(1)}\|_2^2 + 2 \langle Ax_{(1)}, A x_{(\flat)} \rangle + \|A x_{(\flat)}\|_2^2.
		\label{eq:inner_products_order1}
	\end{align}
	The first term on the right hand side, $\|A x_{(1)}\|_2^2$ lies within $(1 \pm \epsilon) \|x\|_2^2$ by the RIP, for the second term we obtain
	\[
	| \langle Ax_{(1)}, A x_{(\flat)} \rangle | = \left| \sum_{\substack{k \in \supp(x_{(1)}) \\ l \in \supp(x_{(\flat)})}} \xi_k \xi_l B_{k, l} \right| =
	\left| \sum_{l \in \supp(x_{(\flat)}) } \xi_{l} a_l \right|
	\]
	for $a_l := \sum_{k \in \supp(x_{(1)})} \xi_k B_{k, l}$. Note that $(\xi_l)_{l \in \supp(x_{(\flat)})}$ and $(\xi_k)_{k \in \supp(x_{(1)})}$ are independent. By conditioning on $(\xi_k)_{k \in \supp(x_{(1)})}$ (and thus $(a_l)_{l \in \supp(x_{(\flat)})}$), the expression $\sum_{l \in \supp(x_{(\flat)}) } \xi_{l} a_l$ becomes a Rademacher chaos of order $1$ and can be controlled with Hoeffding's inequality which corresponds to the general Rademacher chaos concentration inequality for order $1$.
	
	The third term in (\ref{eq:inner_products_order1}) then can be controlled using the Hanson-Wright inequality while the corresponding bound in \cite{riptojl} requires the largest $s$ entries to be considered separately.
	
	Thus, the analysis of the first two terms in (\ref{eq:inner_products_order1}) uses specific properties of Rademacher vectors while the analysis of the third term could be done for a Gaussian vector in essentially the same way. This is indeed necessary since the Johnson-Lindenstrauss property of $P_\Omega H D_\xi$ would not hold if we replaced the vector $\xi$ with a Gaussian vector $g$: Considering the set $E = \{e_1, \dots, e_N\}$ of all canonical basis vectors, $\|P_\Omega H D_g e_k\|_2 = g_k$ for $1 \leq k \leq N$ and thus $\mathbb{E} \max_{k \in [N]} \|P_\Omega H D_g e_k\|_2 \sim \sqrt{\log(N)}$, i.e., for large $N$, we cannot expect all $\|P_\Omega H D_g e_k\|_2$ to be in the range of $1 \pm \epsilon$.
	
	For the case of general $d$, a data point is a multidimensional array $\bd{x} \in \R^{n_1 \times \dots \times n_d}$ and analogously to the separation of the largest $s$ entries in the previous case, we will separate the largest entries for every possible combination of the $d$ axes. Specifically, for any subset $S \subset [d]$, we define the array $\bd{x}^{(S)}$ in such a way that for each fixed value of $(\bd{i}_k)_{k \in S^c}$, we pick the $s^{|S|}$ values of $(\bd{i}_k)_{k \in S}$ with the largest $|x_{\bd{i}_1, \dots, \bd{i}_d}|$. For these indices $\bd{i}$, $x^{(S)}_{\bd{i}} = x_{\bd{i}}$ and otherwise $x^{(S)}_{\bd{i}} = 0$.
	
	For the simplest case $d = 1$, this leads to the following result:
	\begin{itemize}
		\item $\bd{x}^{(\emptyset)}$ is the complete array $\bd{x}$.
		\item $\bd{x}^{(\{1\})}$ contains the largest $s^1$ entries of the entire array.
	\end{itemize}
	
	In an additional step, we set more entries to $0$ in such a way that every non-zero entry of $\bd{x}$ is contained in $\bd{x}^{(S)}$ for at most (and thus precisely) one $S$, specifically such an $S$ of maximal cardinality is chosen. In the case of $d = 1$, this means that all the non-zero entries of $\bd{x}^{(\{1\})}$ are removed from $\bd{x}^{(\emptyset)}$ (set to $0$) and this exactly corresponds to the split of $x = x_{(1)} + x_{(\flat)}$ as described above.
	
	In the case $d = 2$, $\bd{x}$ can be regarded as a matrix with rows and columns and one obtains the following arrays (before removing more entries as described above):
	\begin{itemize}
		\item $\bd{x}^{(\emptyset)}$ is the complete array $\bd{x}$.
		\item $\bd{x}^{(\{1\})}$ contains the largest $s$ entries of every column.
		\item $\bd{x}^{(\{2\})}$ contains the largest $s$ entries of every row.
		\item $\bd{x}^{(\{1, 2\})}$ contains the largest $s^2$ entries of the entire array.
	\end{itemize}
	
	And this principle can be generalized to arbitrary orders $d$. With the aforementioned modification such that every entry is just contained in one $\bd{x}^{(S)}$, it holds that 
	\[
	\bd{x} = \sum_{S \subset [d]} \bd{x}^{(S)}.
	\]
	
	With $\vect: \R^{n_1 \times \dots \times n_d} \rightarrow \R^N$ being the function that maps a multidimensional array to a vector, the norm of interest is
	\[
	\|A \vect(\bd{x})\|_2^2 = \langle A \vect(\bd{x}), A \vect(\bd{x}) \rangle =
	\sum_{S, T \subset [d]} \langle A \vect(\bd{x}^{(S)}), A \vect(\bd{x}^{(T)}) \rangle.
	\]
	The approach is to control $\langle A \vect(\bd{x}^{(S)}), A \vect(\bd{x}^{(T)}) \rangle$ by conditioning on the values of the Rademacher vectors associated to the axes in $S$ on the left hand side of the inner product and along the axes in $J$ on the right hand side. Then this becomes a Rademacher chaos of order $2d - |S| - |T|$ and this can be controlled using Theorem \ref{thm:subgaussian_decoupled}. By adding up the $L_p$ norms for the deviations of all these inner products from their expectation over $S$ and $T$, we obtain an upper bound for the $L_p$ norm of the deviation of the entire norm $\|A \vect(\bd{x})\|_2^2$ as desired.
	
	\begin{figure}[h]
		\vspace{-0.5cm}
		\begin{center}
			\def\dsx{0.45}
			\def\dsy{0.35}
			\begin{tikzpicture}
				\draw (-1.25, 0) -- (5, 0);
				
				\foreach \y in {-1, 1}
				\foreach \x in {1,2,...,10}
				\node at (\x * \dsx, \y * \dsy) [circle, fill, inner sep=1.5pt]{};
				
				\draw plot [smooth cycle] coordinates {(0.8 * \dsx, 0.5 * \dsy) (2.2 * \dsx, 0.5 * \dsy) (2.2 * \dsx, 1.5 * \dsy) (0.8 * \dsx, 1.5 * \dsy)};
				\draw plot [smooth cycle] coordinates {(0.8 * \dsx, -0.5 * \dsy) (3.2 * \dsx, -0.5 * \dsy) (3.2 * \dsx, -1.5 * \dsy) (0.8 * \dsx, -1.5 * \dsy)};
				
				\draw[gray, dashed] plot [smooth cycle] coordinates {(2.8 * \dsx, 0.7 * \dsy) (4.2 * \dsx, 0.7 * \dsy) (4.2 * \dsx, 1.3 * \dsy) (2.8 * \dsx, 1.3 * \dsy)};
				\draw[gray, dashed] plot [smooth cycle] coordinates {(4.8 * \dsx, 0.7 * \dsy) (7.2 * \dsx, 0.7 * \dsy) (7.2 * \dsx, 1.3 * \dsy) (4.8 * \dsx, 1.3 * \dsy)};
				\draw[gray, dashed] plot [smooth cycle] coordinates {(3.8 * \dsx, -0.7 * \dsy) (4.2 * \dsx, -0.7 * \dsy) (4.2 * \dsx, -1.3 * \dsy) (3.8 * \dsx, -1.3 * \dsy)};
				\draw[gray, dashed] plot [smooth cycle] coordinates {(4.8 * \dsx, -0.7 * \dsy) (6.2 * \dsx, -0.7 * \dsy) (6.2 * \dsx, -1.3 * \dsy) (4.8 * \dsx, -1.3 * \dsy)};
				
				\draw[gray, dashed] plot [smooth cycle] coordinates {(7.8 * \dsx, -0.7 * \dsy) (7.8 * \dsx, 1.3 * \dsy) (8.2 * \dsx, 1.3 * \dsy) (8.2 * \dsx, -0.7 * \dsy) (9.2 * \dsx, -0.7 * \dsy) (9.2 * \dsx, -1.3 * \dsy) (6.8 * \dsx, -1.3 * \dsy) (6.8 * \dsx, -0.7 * \dsy)};
				\draw[gray, dashed] plot [smooth cycle] coordinates {(8.8 * \dsx, 0.7 * \dsy) (8.8 * \dsx, 1.3 * \dsy) (10.2 * \dsx, 1.3 * \dsy) (10.2 * \dsx, -1.3 * \dsy) (9.8 * \dsx, -1.3 * \dsy) (9.8 * \dsx, 0.7 * \dsy)};

				\draw plot [smooth cycle] coordinates {(2.9 * \dsx, 0.5 * \dsy) (2.9 * \dsx, 1.5 * \dsy) (7.1 * \dsx, 1.5 * \dsy) (7.1 * \dsx, 0.5 * \dsy)};
				\draw plot [smooth cycle] coordinates {(3.9 * \dsx, -0.5 * \dsy) (3.9 * \dsx, -1.5 * \dsy) (6.2 * \dsx, -1.5 * \dsy) (6.2 * \dsx, -0.5 * \dsy)};
				\draw plot [smooth cycle] coordinates {(7.8 * \dsx, 1.5 * \dsy) (10.2 * \dsx, 1.5 * \dsy) (10.2 * \dsx, -1.5 * \dsy) (6.8 * \dsx, -1.5 * \dsy) (6.8 * \dsx, -0.5 * \dsy) (7.5 * \dsx, -0.5 * \dsy)};
				
				\node at (-1.5 * \dsx, 1 * \dsy) {$[d]:$};
				\node at (-1.5 * \dsx, -1 * \dsy) {$[2 d] \backslash d:$};
				
				\node at (1.5 * \dsx, 2.5 * \dsy) {$S$};
				\node at (2 * \dsx, -2.5 * \dsy) {$T$};
				\node at (5 * \dsx, 2.5 * \dsy) {$\bar{I}$};
				\node at (5 * \dsx, -2.5 * \dsy) {$\bar{I}'$};
				\node at (11 * \dsx, 1 * \dsy) {$\bar{J}$};
			\end{tikzpicture}
		\end{center}
		\vspace{-0.5cm}
		\label{fig:partition}
		\caption{A sketch showing an example of the sets $\bar{I},\, \bar{I}',\, \bar{J}$: The axes (black dots) are divided into $[d]$ and $[2d] \backslash [d]$ (horizontal line) depending on whether they arise on the left or right hand side of the inner product. The partition $\bar{I}_1, \dots, \bar{I}_\kappa$ of $[2 d] \backslash (S \cup T + d)$ is shown by the gray dashed lines. The partition sets which only intersect one of $[d]$ or $[2 d] \backslash [d]$ are joined to $\bar{I}$ and $\bar{I}'$ respectively. The remaining partition sets intersect both sides and are joined to $\bar{J}$.}
	\end{figure}
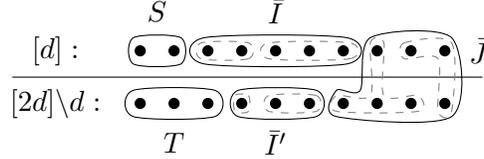
	
	So for any $S, T \subset [d]$, we need to control the $\| \cdot \|_{\bar{I}_1, \dots, \bar{I}_\kappa}$ norm of the array corresponding to the chaos of order $2d - |S| - |T|$ where $\bar{I}_1, \dots, \bar{I}_\kappa$ is an arbitrary partition of the indices $(S^c) \cup (T^c + d)$. Now among the sets $\bar{I}_1, \dots, \bar{I}_\kappa$, we distinguish three types: Those which ones are contained in $[d]$, the ones contained in $[2 d] \backslash [d]$, and those which intersect both these sets. For each type, we form the union of the corresponding sets and obtain the sets $\bar{I}$, $\bar{I}'$ and $\bar{J}$ respectively (see Figure~\ref{fig:partition} for an example). By the definition of the norms in (\ref{eq:partition_norm}), such a norm cannot become smaller by joining some of the partition sets and thus
	\[
	\| \cdot \|_{\bar{I}_1, \dots, \bar{I}_\kappa} \leq \| \cdot \|_{\bar{I}, \bar{I}', \bar{J}}
	\]
	and it is sufficient to find a suitable upper bound on the norm on the right hand side.
	
	So we can rearrange $\bd{x}^{(S)}$ as an array $(x_{j, k, l})_{j, k, l}$ of order $3$ according to the partition of $[d]$ into $S$, $\bar{I}$ and the remaining axes. In an analogous way, $\bd{x}^{(T)}$ can be rearranged to $(y_{j', k', l'})_{j', k', l'}$ according to $T$, $\bar{I}' - d$ and the remaining axes. Now rearrange $\Phi^* \Phi$ to $(B_{j, k, l, j', k', l'})_{j, k, l, j', k', l'}$ according to the arrangement of $\bd{x}^{(S)}$ and $\bd{x}^{(T)}$ to three axes each. We also rearrange $\bigotimes_{r \in S} \xi^{(r)}$ as $(\xi_j)_{j}$ and define $(\xi'_{j'})_{j'}$ analogously. Then the $\| \cdot \|_{\bar{I}, \bar{I}', \bar{J}}$ norm of interest is
	\[
	\sup_{\|\alpha\|_2 = \|\alpha'\|_2 = \|\beta\|_2 = 1} \sum_{\substack{j, k, l \\ j', k', l'}} B_{j, k, l, j', k', l'} \xi_j \xi_{j'}' x_{j, k, l} y_{j', k', l'} \alpha_k \alpha_{k'}' \beta_{j, j'}.
	\]
	
	Noting that for any $z \in \R^n$, $\sup_{\alpha \in \R^n, \|\alpha\|_2 = 1} \sum_{k = 1}^n z_k \alpha_k = \left(\sum_{k = 1}^n z_k^2\right)^{1/2}$, the above expression becomes
	\[
	\sup_{\|\alpha\|_2 = \|\alpha'\|_2 = 1} \left[ \sum_{j, j'} \left( \sum_{k, l, k', l'} B_{j, k, l, j', k', l'} \xi_j \xi_{j'}' x_{j, k, l} y_{j', k', l'} \alpha_k \alpha_{k'}' \right)^2 \right]^{1/2}.
	\]
	
	This expression can be controlled using Lemma \ref{lem:inner_pro_bound}. The vectors $u^{(j, k)}$ of that lemma are vectorized versions of the order $3$ array $\bd{x}$ in which the non-zero entries are restricted to one slice (one particular value of $(j, k)$), multiplied with the corresponding $\xi_k$. The $v^{(j', k')}$ are obtained in an analogous way. In result, this lemma can bound the respective $\| \cdot \|_{\bar{I}, \bar{I}', \bar{J}}$ norm by $C \cdot \frac{\delta} {s^{\frac{|\bar{J}|}{4} + \frac{|\bar{I}|}{2} + \frac{|\bar{I}'|}{2}}}$. The combinatorial Lemma \ref{lem:partition_counting} now ensures that in any case, $\frac{|\bar{J}|}{4} + \frac{|\bar{I}|}{2} + \frac{|\bar{I}'|}{2} \geq \frac{\kappa}{2}$ where $\kappa$ is the number of sets in the original partition. So altogether, we can bound the $\| \cdot \|_{\bar{I}_1, \dots, \bar{I}_\kappa}$ norm by $C \cdot \frac{\delta}{s^\frac{\kappa}{2}}$.
	
	Finally, by Theorem \ref{thm:chaos_moments}, we can use this to bound the $L_p$ moments of the deviations of all inner products $\langle A \vect(\bd{x}^{(S)}), A \vect(\bd{x}^{(T)}) \rangle$ and obtain
	\[
	\left\| \| A \vect(\bd{x}) \|_2^2 - \|\bd{x}\|_2^2 \right\|_{L_p} \leq C(d) \sum_{k = 1}^{2 d} p^\frac{k}{2} \frac{\delta}{s^\frac{k}{2}}.
	\]
	
	This can be converted to a concentration inequality using Lemma \ref{lem:moment_tail_bound}, implying the distributional Johnson-Lindenstrauss property.
	
	\begin{remark} \label{rem:oblivious_sketching}
		In \cite{ahle2020oblivious}, Lemma 4.11 (TensorSRHT) provides a Johnson-Lindenstrauss result which is similar to our Corollary \ref{cor:jl_hadamard} restricted to Hadamard matrices. The proof of this lemma can be found in the extended version \cite{ahle2020obliviousArxiv}. Their proof uses general moment bounds for sums of independent mean $0$ variables to control the probability in the subsampling $P_\Omega$ while conditioning on the random sign vector $\xi$. In contrast, our approach conditions on the RIP of $P_\Omega H$ and then shows the Johnson-Lindenstrauss property by controlling the probability in $\xi$. This gives an advantage for the case that the Johnson-Lindenstrauss property is shown for $p$ vectors simultaneously. For our approach, once $P_\Omega H$ has the RIP, this holds for all $s$-sparse vectors uniformly. Then we only need to show the Johnson-Lindenstrauss property by a union bound with respect to the probability in $\xi$ but not with respect to $P_\Omega$. The advantage of this is that in this case the dependence of the embedding dimension in \cite{ahle2020oblivious} is $(\log p)^{d+1}$ (up to smaller logarithmic factors) while our result only requires $(\log p)^d$ which the example in Section \ref{sec:lower_bounds} proves to be optimal.
		
		On the other hand, our approach makes controlling the probability in $\xi$ more intricate. In \cite{ahle2020oblivious}, Lemma 4.9 provides a result similar to the one by Latala \cite{latala_gauss_chaos} with a better dependence on $d$ but all $\|\cdot\|_{I_1, \dots, I_\kappa}$ bounded by the Frobenius norm. This suffices to control $\xi^T (D_{H_j} x)$ sufficiently for arbitrary $x \in \R^N$ where $H_j$ is the $j$-th row of the Hadamard matrix. The latter is required in \cite{ahle2020oblivious}. In our case, we need to control $\xi^T D_{x} \Phi^T \Phi D_{x} \xi$ for which we make use of the RIP of $\Phi$ and control all the $\|\cdot\|_{I_1, \dots, I_\kappa}$ norms separately. We will discuss more aspects of the relation of our work to \cite{ahle2020oblivious} in Section \ref{sec:discussion}.

	\end{remark}

	\section{Main proofs} \label{sec:proof_general}

	\subsection{Lemma for RIP matrices}

	In this subsection, we prove the following technical lemma mentioned in the previous section. It makes use of the restricted isometry property of the considered matrix and will be of central importance in the proof of Theorem~\ref{thm:main} in Subsection~\ref{sec:main_prof_subsec}. However, in this form it is independent of the underlying higher-order array structures of the signals.
	
	\begin{lemma} \label{lem:inner_pro_bound}
		Let $n_1,\,n_2,\,n_3,\,n_1',\,n_2',\,n_3'$ and $s_1,\,s_2,\,s_3,\,s_1',\,s_2',\,s_3'$ be natural numbers such that $N = n_1 n_2 n_3 = n_1' n_2' n_3'$ and $s = s_1 s_2 s_3 = s_1' s_2' s_3'$ and $\Phi \in \R^{m \times N}$ has the $(2 s, \delta)$-RIP, $B := \Phi^* \Phi - Id_{N}$. Consider vectors $u^{(j, k)}, v^{(j', k')} \in \R^{n_1 n_2 n_3}$ for $(j, k, j', k') \in [n_1] \times [n_2] \times [n_1'] \times [n_2']$ such that all $u^{(j, k)}$ are $s_3$-sparse with disjoint supports and all $v^{(j', k')}$ are $s_3'$-sparse with disjoint supports.
		
		Let $(b_2(1), \dots, b_2(R_2))$ be a partition of $[n_2]$ into sets of size $\leq s_2$ each where $R_2 = \lceil \frac{n_2}{s_2} \rceil$.
		For each $K \in [R_2]$, let $(b_1^{(K)}(1), \dots, b_1^{(K)}(R_1))$ be a partition of $[n_1]$ into sets of size $\leq s_1$ each where $R_1 = \lceil \frac{n_1}{s_1} \rceil$.
		
		
		Analogously, let $b_2'(1), \dots, b_2'(R_2')$ and $b_1^{(K')'}(1), \dots, b_1^{(K')'}(R_1')$ be partitions of $[n_2']$ into sets of size $\leq s_2'$ and $[n_1']$ into sets of size $\leq s_1'$, respectively.
		
		Also, for $j \in [n_1]$ and $K \in [R_2]$, define $u^{(j, (K))} := \sum_{k \in b_2(K)} u^{(j, k)}$ and define $v^{(j', (K'))}$ in the analogous way.
		
		Then
		\begin{align*}
			& \sum_{(j, j') \in [n_1] \times [n_1']} \left( \sum_{(k, k') \in [n_2] \times [n_2']} (u^{(j, k)})^* B v^{(j', k')} \right)^2 \\
			\leq & \delta^2 \left[ \sum_{(J, K, \bar{K}) \in [R_1] \times [R_2]^2} \sqrt{  \sum_{j \in b^{(K)}_{1}(J)} \|u^{(j, (K))}\|_2^2 \|u^{(j, \bar{K})}\|_2^2} \right] \\
			& \cdot \left[ \sum_{(J', K', \bar{K}') \in [R_1'] \times [R_2']^2} \sqrt{ \sum_{j' \in b^{(K')'}_{1'}(J')} \|v^{(j', (K'))}\|_2^2 \|v^{(j', (\bar{K}'))}\|_2^2 } \right].
		\end{align*}
	\end{lemma}
	
	\begin{proof}
		
		We can bound the desired expression by rearranging terms and blockwise summation.
		
		\begin{align*}
			& \sum_{(j, j') \in [n_1] \times [n_1']} \left( \sum_{(k, k') \in [n_2] \times [n_2']} (u^{(j, k)})^* B v^{(j', k')} \right)^2 \\
			= & \sum_{(j, j') \in [n_1] \times [n_1']} \sum_{(k, \bar{k}, k', \bar{k}') \in [n_2]^2 \times [n_2']^2} (u^{(j, k)})^* B v^{(j', k')} (v^{(j', \bar{k}')})^* B^* u^{(j, \bar{k})} \\
			= & \sum_{\substack{(J, J') \in [R_1] \times [R_1'] \\ (K, \bar{K}) \in [R_2]^2 \\ (K', \bar{K}') \in [R_2']^2 \\ j \in b_1^{(K)}(J) }}  (u^{(j, (K))} )^* B \left( \sum_{j' \in b^{(K')}_{1'}(J')} v^{(j', (K'))} (v^{(j', (\bar{K}'))})^* \right) B^* u^{(j, \bar{K})} \\
			= & \sum_{\substack{J \in [R_1] \\ J' \in [R_1'] \\ (K, \bar{K}) \in [R_2]^2 \\ (K', \bar{K}') \in [R_2']^2}} \left\langle  \sum_{\substack{j \in \\ b_{1}^{(K)}(J)}} u^{(j, (K))} (u^{(j, \bar{K})})^*, B \left( \sum_{\substack{j' \in \\ b_{1'}^{(K')}(J')}} v^{(j', (K'))} (v^{(j', (\bar{K}'))})^* \right) B^* \right\rangle_F.
		\end{align*}
		
		Note that every $u^{(j, k)}$ is $s_3$-sparse and $|b_2(K)| \leq s_2$, $|b_1^{(K)}(J)| \leq s_1$. Thus, the number of nonzero rows and the number of nonzero columns of the matrix $\sum_{j \in b_{1}^{(K)}(J)} u^{(j, (K))} (u^{(j, \bar{K})})^*$ can be at most $s = s_1 s_2 s_3$ each. The same holds for $\sum_{j' \in b_{1'}^{(K')'}(J')} v^{(j', (K'))} (v^{(j', (\bar{K}'))})^*$. So we can restrict $B$ to a submatrix of $s$ rows and $s$ columns which has an operator norm $\leq \delta$ by the RIP assumption (Lemma \ref{lem:rip_spec_norm}). Using that $\langle A, B C B^* \rangle_F \leq \|A\|_F \|B C B^*\|_F \leq \|A\|_F \|B\|_{2 \rightarrow 2}^2 \|C\|_F$, we can bound the expression by
		
		\begin{align*}
			& \delta^2 \sum_{\substack{J \in [R_1] \\ J' \in [R_1'] \\ (K, \bar{K}) \in [R_2]^2 \\ (K', \bar{K}') \in [R_2']^2}} \left\|  \sum_{j \in b^{(K)}_{1}(J)} u^{(j, (K))} (u^{(j, \bar{K})})^* \right\|_F \left\| \sum_{j' \in b_{1'}^{(K')}(J')} v^{(j', (K'))} (v^{(j', (\bar{K}'))})^* \right\|_F \\
			\leq & \delta^2\sum_{\substack{(J, J') \in [R_1] \times [R_1'] \\ (K, \bar{K}) \in [R_2]^2 \\ (K', \bar{K}') \in [R_2']^2}} \sqrt{  \sum_{j \in b_{1}^{(K)}(J)} \|u^{(j, (K))}\|_2^2 \|u^{(j, \bar{K})}\|_2^2 \sum_{j' \in b_{1'}^{(K')}(J')} \|v^{(j', (K'))}\|_2^2 \|v^{(j', (\bar{K}'))}\|_2^2 } \\
			= & \delta^2 \left[ \sum_{\substack{J \in [R_1] \\ (K, \bar{K}) \in [R_2]^2}} \sqrt{  \sum_{j \in b^{(K)}_{1}(J)} \|u^{(j, (K))}\|_2^2 \|u^{(j, \bar{K})}\|_2^2} \right] \\
			& \cdot \left[ \sum_{\substack{J' \in [R_1'] \\ (K', \bar{K}') \in [R_2']^2}} \sqrt{ \sum_{j' \in b^{(K')}_{1'}(J')} \|v^{(j', (K'))}\|_2^2 \|v^{(j', (\bar{K}'))}\|_2^2 } \right],
		\end{align*}
		where in the first step we used that the $u^{(j, (K))} (u^{(j, \bar{K})})^*$ have disjoint supports.
		
	\end{proof}

	\subsection{Notation of array indices} \label{sec:array_indices}
	
	In this subsection, we introduce the notation used to handle higher order arrays. In most parts, this is adapted from \cite{hanson_wright_tensors}, see Section 1.4 there.
	
	Consider a vector of dimensions $\bd{n} = (n_1, n_2, \dots, n_d)$ and $I \subset [d]$. We call a function $\bd{i}: I \rightarrow \mathbb{N}$ a partial index of order $d$ on $I$ if for all $l \in I$, $\bd{i}_l := \bd{i}(l) \in [n_l]$. Assume there is exactly one such function if $I = \emptyset$. If $I = [d]$, then $\bd{i}$ is called an index of order $d$. We denote the set of all partial indices of order $d$ on $I$ as $\Jn(I)$ and $\Jn := \Jn([d])$ is the set of all indices of order $d$. $\Jn$ can be identified with $[n_1] \times \dots \times [n_d]$. We also write $\bd{n}^{\times 2} = (n_1, \dots, n_d, n_1, \dots, n_d)$.
	
	A function $\bd{x}: \Jn \rightarrow \R$ is called an array of order $d$. Because of the aforementioned identification, we also write $\bd{x} \in \R^{n_1 \times \dots \times n_d} =: \R^{\bd{n}}$. Similarly $\bar{\bd{x}}: \Jn(I) \rightarrow \R$ is called a partial array and we denote $\bar{\bd{x}} \in \R^{\bd{n}}(I)$. For $I = [d]$, this is just a regular array. We denote
	\[
	\|\bar{\bd{x}}\|_2 := \left[ \sum_{\bd{i} \in \Jn(I)} \bar{x}_{\bd{i}}^2 \right]^\frac{1}{2}.
	\]
	
	For partial Indices $\bd{i}$ on $I$ and $\bd{j}$ on $J$, denote $\bd{i} \dot{\times} \bd{j}$ for the partial index on $I \cup J$ such that for $l \in I \cup J$,
	\[
	(\bd{i} \dot{\times} \bd{j})_l = \begin{cases}
		\bd{i}_l & \text{if } l \in I \\
		\bd{j}_l & \text{if } l \in J
	\end{cases}
	\]
	
	Similarly, for $I \subset [2 d]$,  $J \subset [d]$ such that $I \cap (J + d) = \emptyset$, $\bd{i} \in \bd{J}^{\bd{n}^{\times 2}}(I)$ and $\bd{j} \in \bd{J}^{\bd{n}}(J)$, we define the index $\bd{i} \dot+ \bd{j} \in \bd{J}^{\bd{n}^{\times 2}}(I \cup (J + d))$ such that $\bd{j}$ is shifted by $d$ in the sense that
	\[
	(\bd{i} \dot+ \bd{j})_{l} = \begin{cases}
		\bd{i}_l & \text{if } l \in I \\
		\bd{j}_{l - d} & \text{if } l \in J + d.
	\end{cases}
	\]
	
	The following function establishes a relation between array indices and indices of the rearrangement of the array as a vector.
	
	\begin{definition}
		For a dimension vector $\bd{n} = (n_1, n_2, \dots, n_d)$, a subset $I = \{j_1, \dots, j_{|I|}\} \subset [d]$ for $j_1 < \dots < j_{|I|}$ and $N := \prod_{l \in I} n_l$, define the function $L^\bd{n}_I: \Jn(I) \rightarrow [N]$ by
		\[
		L^\bd{n}_I(\bd{j}) = \sum_{l = 1}^{|I|} (\bd{i}_{j_l} - 1) \prod_{l' = 1}^{l - 1} n_{j_{l'}} + 1.
		\]
		which defines a bijection. Its inverse is called $\hat{L}^\bd{n}_I: [N] \rightarrow \Jn(I)$.
	\end{definition}
	
	\begin{definition}
		For an array $\bd{a} \in \R^{\bd{n}}$, the vectorization $\vect(\bd{a}) \in \R^N$ is defined such that for all $\bd{j} \in \Jn$, $(\vect(\bd{a}))_{L^\bd{n}(\bd{j})} = \bd{a}_{\bd{j}}$.
	\end{definition}
	
	\begin{definition}
		Let $I \subset J \subset [d]$. For a partial index $\bd{j} \in \Jn(J)$, define the restriction $\bd{j}_I \in \Jn(I)$ such that for all $l \in I$, $(\bd{j}_I)_l = \bd{j}_{l}$.
	\end{definition}
	
	Usually, we denote vectors and matrices in regular letters and arrays of general order in bold letters. In some cases, we rearrange the entries of vectors or matrices to higher order arrays. In these cases, we use matching letters, e.g., $A$ for the matrix and $\bd{A}$ for its rearrangement as a higher order array.

	\subsection{Decoupling} \label{sec:decoupling}
	
	Similar to the well-know decompling lemma for chaoses of order 2 (see e.g.~Theorem 6.1.1 in \cite{vershynin_hdp}), a main ingredient for the proof of our main result will be a decoupling theorem for our higher-order case. Such a decoupling theorem for exactly the setup considered in this work has been studied in \cite{hanson_wright_tensors} where Theorem 2.5 together with the subsequent Remark 2.6 about its simplification for the Rademacher case yields the following result.
	
	\begin{theorem}[Theorem 2.5, Remark 2.6 in \cite{hanson_wright_tensors}] \label{thm:decoupling_main}
		Let $\bd{n} = (n_1, \dots, n_d) \in \mathbb{N}^d$, $\bd{A} \in \mathbb{R}^{\bd{n}^{\times 2}}$, $\xi^{(1)} \in \{\pm 1\}^{n_1}, \dots, \xi^{(d)} \in \{\pm 1\}^{n_d}$ independent Rademacher vectors and $\bar{\xi}^{(1)}, \dots, \bar{\xi}^{(d)}$ corresponding independent copies. Then
		\begin{align*}
			& \left\|\sum_{\bd{i}, \bd{i}' \in \Jn} A_{\bd{i} \dot+ \bd{i}'} \prod_{l \in [d]} \xi^{(l)}_{\bd{i}_l} \xi^{(l)}_{\bd{i}'_l} - \mathbb{E} \sum_{\bd{i}, \bd{i}' \in \Jn} A_{\bd{i} \dot+ \bd{i}'} \prod_{l \in [d]} \xi^{(l)}_{\bd{i}_l} \xi^{(l)}_{\bd{i}'_l} \right\|_{L_p} \\
			& \leq
			\sum_{\substack{I \subset [d]: \\ I \neq [d]}}
			4^{d - |I|} \left\| \sum_{\substack{\bd{i} \in \Jn(I) \\ \bd{k}, \bd{k}' \in \Jn(I^c)}} A_{(\bd{i} \dot\times \bd{k}) \dot+ (\bd{i} \dot\times \bd{k}')} \prod_{l \in I^c} \xi^{(l)}_{\bd{k}_l} \bar{\xi}^{(l)}_{\bd{k}'_l} \right\|_{L_p}
		\end{align*}
	\end{theorem}
	
	In the proof of our main result, we will consider such higher order chaos expressions for arrays $\bd{A} \in \R^{\bd{n}^{\times 2}}$ such that
	\[
	A_{\bd{i} \dot+ \bd{i}'} = B_{\bd{i} \dot+ \bd{i}'} x_{\bd{i}} x_{\bd{i}'}
	\]
	for all $\bd{i}, \bd{i}' \in \Jn$, where $\bd{B} \in \R^{\bd{n}^{\times 2}}$ and $\bd{x} \in \R^{\bd{n}}$.
	Specifically, we need to obtain norm bounds as in Theorem \ref{thm:decoupling_main} that hold for one particular choice of $\bd{B}$ but all $\bd{x} \in \R^{\bd{n}}$. In this situation, the following theorem refines the statement of Theorem \ref{thm:decoupling_main}.
	
	\begin{theorem} \label{thm:decoupling_2}
		Let $\bd{B} \in \mathbb{R}^{\bd{n}^{\times 2}}$ and $\xi^{(1)} \in \{\pm 1\}^{n_1}, \dots, \xi^{(d)} \in \{\pm 1\}^{n_d}$ be independent Rademacher vectors. Define $\bm{\xi} \in \R^{\bd{n}}$ by
		\[
		\xi_{\bd{i}} = \prod_{l = 1}^d \xi^{(l)}_{\bd{i}_l}.
		\]
		and let $\bar{\bm{\xi}}$ be an independent copy of $\bm{\xi}$. Let $p \geq 1$.
		
		Assume that for all $\bd{x} \in \R^{\bd{n}}$, $\|\bd{x}\|_2 = 1$, it holds that
		\[
		\left\| \sum_{\bd{i}, \bd{i}' \in \Jn} B_{\bd{i} \dot+ \bd{i}'} x_{\bd{i}} x_{\bd{i}'} \xi_{\bd{i}} \bar{\xi}_{\bd{i}'} \right\|_{L_p} \leq \gamma_p.
		\]
		
		Then also for all $\bd{x} \in \R^{\bd{n}}$, $\|\bd{x}\|_2 = 1$, it holds that
		\[
		\left\| \sum_{\bd{i}, \bd{i}' \in \Jn} B_{\bd{i} \dot+ \bd{i}'} x_{\bd{i}} x_{\bd{i}'} \xi_{\bd{i}} \xi_{\bd{i}'} -
		\mathbb{E} \sum_{\bd{i}, \bd{i}' \in \Jn} B_{\bd{i} \dot+ \bd{i}'} x_{\bd{i}} x_{\bd{i}'} \xi_{\bd{i}} \xi_{\bd{i}'} \right\|_{L_p} \leq 5^d \gamma_p.
		\]
	\end{theorem}

	\begin{proof}
		By Theorem \ref{thm:decoupling_main}, 
		\begin{equation}
			\left\| \sum_{\bd{i}, \bd{i}' \in \Jn} B_{\bd{i} \dot+ \bd{i}'} x_{\bd{i}} x_{\bd{i}'} \xi_{\bd{i}} \xi_{\bd{i}'} -
			\mathbb{E} \sum_{\bd{i}, \bd{i}' \in \Jn} B_{\bd{i} \dot+ \bd{i}'} x_{\bd{i}} x_{\bd{i}'} \xi_{\bd{i}} \xi_{\bd{i}'} \right\|_{L_p} \leq
			\sum_{\substack{I \subset [d] \\ I \neq [d]}} 4^{d - |I|} \|\bar{S}_I\|_{L_p}
			\label{eq:decoupling_2_first}
		\end{equation}
		
		for
		\[
		\bar{S}_I = \sum_{\bd{i} \in \Jn(I)} \sum_{ \bd{j}, \bd{j}' \in \Jn(I^c) } B_{(\bd{i} \dot\times \bd{j}) \dot+ (\bd{i} \dot\times \bd{j}')} x_{\bd{i} \dot\times \bd{j}} x_{\bd{i} \dot\times \bd{j}'} \prod_{l \in I^c} \xi^{(l)}_{\bd{j}_l} \bar{\xi}^{(l)}_{\bd{j}'_{l}}.
		\]
		
		Now fix $[d] \neq I \subset [d]$ and for each $\bd{i} \in \Jn(I)$, define
		\[
		x^{(\bd{i})}_{\bd{i}' \dot\times \bd{j}} = \begin{cases}
			x_{\bd{i} \dot\times \bd{j}} & \text{if } \bd{i}' = \bd{i} \\
			0 & \text{otherwise}.
		\end{cases}
		\]
		for $\bd{i}' \in \Jn(I)$, $\bd{j} \in \Jn(I^c)$.
		
		This gives us
		\begin{align*}
			\bar{S}_I = &  \sum_{\bar{\bd{i}} \in \Jn(I)} \sum_{\bd{i}, \bd{i}' \in \Jn(I)} \sum_{ \bd{j}, \bd{j}' \in \Jn(I^c) } B_{(\bd{i} \dot\times \bd{j}) \dot+ (\bd{i} \dot\times \bd{j}')} x^{(\bar{\bd{i}})}_{\bd{i} \dot\times \bd{j}} x^{(\bar{\bd{i}})}_{\bd{i}' \dot\times \bd{j}'} \prod_{l \in I^c} \xi^{(l)}_{\bd{j}_l} \bar{\xi}^{(l)}_{\bd{j}'_{l}} \\
			= & \sum_{\bar{\bd{i}} \in \Jn(I)} \sum_{\bd{i}, \bd{i}' \in \Jn} B_{\bd{i} \dot+ \bd{i}'} x^{(\bar{\bd{i}})}_{\bd{i}} x^{(\bar{\bd{i}})}_{\bd{i}'} \prod_{l \in I^c} \xi^{(l)}_{\bd{i}_l} \bar{\xi}^{(l)}_{\bd{i}'_{l}}
		\end{align*}
		
		Note that each summand is non-zero only if $\bd{i}_I = \bd{i}'_I = \bar{\bd{i}}$ and in that case
		\begin{align*}
			\prod_{l \in I^c} \xi^{(l)}_{\bd{i}_l} \bar{\xi}^{(l)}_{\bd{i}'_{l}} \cdot 1 = 
			\prod_{l \in I^c} \xi^{(l)}_{\bd{i}_l} \bar{\xi}^{(l)}_{\bd{i}'_{l}} \cdot \left( \prod_{l \in I} \xi^{(l)}_{\bar{\bd{i}}_l} \bar{\xi}^{(l)}_{\bar{\bd{i}}_{l}} \right)^2 =
			\xi_{\bd{i}} \bar{\xi_{\bd{i}'}} \cdot \prod_{l \in I} \xi^{(l)}_{\bar{\bd{i}}_l} \bar{\xi}^{(l)}_{\bar{\bd{i}}_{l}}.
		\end{align*}
		
		So this yields
		\begin{align*}
			\bar{S}_I = 
			\sum_{\bar{\bd{i}} \in \Jn(I)} \left( \prod_{l \in I} \xi^{(l)}_{\bar{\bd{i}}_l} \bar{\xi}^{(l)}_{\bar{\bd{i}}_l} \right) \sum_{\bd{i}, \bd{i}' \in \Jn} B_{\bd{i} \dot+ \bd{i}'} x^{(\bar{\bd{i}})}_{\bd{i}} x^{(\bar{\bd{i}})}_{\bd{i}'} \xi_{\bd{i}} \bar{\xi_{\bd{i}'}}.
		\end{align*}
		
		Because all the Rademacher variables are $\pm 1$, it holds that $\left| \prod_{l \in I} \xi^{(l)}_{\bar{\bd{i}}_l} \bar{\xi}^{(l)}_{\bar{\bd{i}}_l} \right| = 1$. Using this and the triangle inequality, we obtain
		\begin{align*}
			\|\bar{S}_I\|_{L_p} \leq & \sum_{\bar{\bd{i}} \in \Jn(I)} \left\| \left| \prod_{l \in I} \xi^{(l)}_{\bar{\bd{i}}_l} \bar{\xi}^{(l)}_{\bar{\bd{i}}_l}  \right| \left| \sum_{\bd{i}, \bd{i}' \in \Jn} B_{\bd{i} \dot+ \bd{i}'} x^{(\bar{\bd{i}})}_{\bd{i}} x^{(\bar{\bd{i}})}_{\bd{i}'} \xi_{\bd{i}} \xi_{\bd{i}'} \right| \right\|_{L_p} \\
			= & \sum_{\bar{\bd{i}} \in \Jn(I)} \left\| \sum_{\bd{i}, \bd{i}' \in \Jn} B_{\bd{i} \dot+ \bd{i}'} x^{(\bar{\bd{i}})}_{\bd{i}} x^{(\bar{\bd{i}})}_{\bd{i}'} \xi_{\bd{i}} \xi_{\bd{i}'} \right\|_{L_p} =:
			\sum_{\bar{\bd{i}} \in \Jn(I)} \|\bar{S}_{I, \bar{\bd{i}}} \|_{L_p}.
		\end{align*}
		
		If $\|\bd{x}^{(\bar{\bd{i}})}\|_2 = 0$, then $\|\bar{S}_{I, \bar{\bd{i}}} \|_{L_p} = 0 = \|\bd{x}^{(\bar{\bd{i}})}\|_2^2 \cdot \gamma_p$. Otherwise, the array with entries $\frac{x^{(\bar{\bd{i}})}_{\bd{i}}}{\|\bd{x}^{(\bar{\bd{i}})}\|_2}$ has a $\|\cdot\|_2$ norm of $1$ and by the assumption of the theorem
		\[
		\|\bar{S}_{I, \bar{\bd{i}}} \|_{L_p} \leq \|\bd{x}^{(\bar{\bd{i}})}\|_2^2  \gamma_p
		\]
		which then holds in all cases and this implies
		\[
		\|\bar{S}_I\|_{L_p} \leq \sum_{\bar{\bd{i}} \in \Jn(I)} \|\bd{x}^{(\bar{\bd{i}})}\|_2^2  \gamma_p = \|\bd{x}\|_2^2 \gamma_p = \gamma_p.
		\]
		
		Substituting into (\ref{eq:decoupling_2_first}) yields
		\begin{align*}
			\left\| \sum_{\bd{i}, \bd{i}' \in \Jn} B_{\bd{i} \dot+ \bd{i}'} x_{\bd{i}} x_{\bd{i}'} \xi_{\bd{i}} \xi_{\bd{i}'} -
			\mathbb{E} \sum_{\bd{i}, \bd{i}' \in \Jn} B_{\bd{i} \dot+ \bd{i}'} x_{\bd{i}} x_{\bd{i}'} \xi_{\bd{i}} \xi_{\bd{i}'} \right\|_{L_p} \leq &
			\sum_{\substack{I \subset [d] \\ I \neq [d]}} 4^{d - |I|} \gamma_p \\
			= & \sum_{k = 0}^{d - 1} \sum_{\substack{I \subset [d] \\ |I| = k}} 4^{d - k} \gamma_p.
		\end{align*}
		
		Noting that there are precisely ${d \choose k}$ sets $I \subset [d]$ with $|I| = k$, we can bound this by
		\[
		\gamma_p \sum_{k = 0}^{d - 1} {d \choose k} 4^{d - k} \leq \gamma_p \sum_{k = 0}^{d} {d \choose k} 4^{d - k} \cdot 1^k = 
		\gamma_p (4 + 1)^d = 5^d \gamma_p
		\]
		which completes the proof.
	\end{proof}

	\subsection{Proof of Theorem \ref{thm:main}} \label{sec:main_prof_subsec}
	
	Since every vector $x \in \R^N$ can be rearranged to an array in $\R^{\bd{n}}$, it is sufficient to prove
	\[
	\mathbb{P}\left( \left| \|A \vect(\bd{x})\|_2^2 - \|\bd{x}\|_2^2 \right| > \epsilon \right) \leq \eta
	\]
	for any $\bd{x} \in \R^{\bd{n}}$ with $\|\bd{x}\|_2 = 1$. So take an arbitrary such $\bd{x}$.
	
	\subsubsection{Splitting up $\bd{x}$ into $\bd{x}^{(S)}$}
	
	For every subset $S \subset [d]$, define the set $\bd{K}(S) \subset \Jn$ of indices in the following way: For each $\bd{j} \in \Jn(S^c)$, choose $s^{|S|}$ indices $\bd{i} \in \Jn(S)$ with the largest $|x_{\bd{i} \dot\times \bd{j}}|$ and $\bd{K}(S)$ is the set of all $\bd{i} \dot\times \bd{j}$ obtained in this way.
	
	Now for every index $\bd{i} \in \Jn$, choose $S(\bd{i})$ to be a set $S \subset [d]$ of largest cardinality such that $\bd{i} \in \bd{K}(S)$. Since $\bd{K}(\emptyset) = \Jn$, such an $S$ always exists.
	
	Then for any set $S \subset [d]$, define $\bd{x}^{(S)} \in \R^{\bd{n}}$ such that for each $\bd{i} \in \Jn$,
	\[
	x^{(S)}_{\bd{i}} := \begin{cases}
		x_{\bd{i}} & \text{if } S(\bd{i}) = S \\
		0 & \text{otherwise}.
	\end{cases}
	\]
	
	Since for every index $\bd{i} \in \Jn$, we chose exactly one $S(\bd{i})$,
	\[
	\bd{x} = \sum_{S \subset [d]} \bd{x}^{(S)}.
	\]
	
	A direct consequence from these definitions is the following lemma.
	
	\begin{lemma} \label{lem:non_0_entries}
		Let $S \subset [d]$. For any index $\bd{j} \in \Jn(S^c)$, there can be at most $2^{|S|}$ different indices $\bd{i} \in \Jn(S)$ such that
		\[
		x^{(S)}_{\bd{i} \dot\times \bd{j}} \neq 0.
		\]
	\end{lemma}
	
	\begin{proof}
		Fix $\bd{j} \in \Jn(S^c)$.
		
		If $x^{(S)}_{\bd{i} \dot\times \bd{j}} \neq 0$ holds for $\bd{i} \in \Jn(S)$, then $S(\bd{i} \dot\times \bd{j}) = S$, implying that $\bd{i} \dot\times \bd{j} \in \bd{K}(S)$. By definition of $\bd{K}(S)$ however, this can only be the case for $s^{|S|}$ different indices $\bd{i} \in \Jn(S)$.
	\end{proof}

	\begin{lemma} \label{lem:max_sum_inequality_1}
		For any $S, T \subset [d]$ with $|S| < |T|$ and any $\bd{k} \in \Jn(T^c)$,
		\[
		\max_{\bd{j} \in \Jn(T)} |x^{(S)}_{\bd{j} \dot\times \bd{k}}|^2 \leq \frac{1}{s^{|T|}} \sum_{\bd{j} \in \Jn(T)} |x_{\bd{j} \dot\times \bd{k}}|^2.
		\]
	\end{lemma}
	
	\begin{proof}
		Let $S, T \subset [d]$, $|S| < |T|$ and $\bd{k} \in \Jn(T^c)$. Choose $\bd{j}_0 \in \Jn(T)$ such that 
		$|x^{(S)}_{\bd{j}_0 \dot\times \bd{k}}|$ is maximal.
		
		If $|x^{(S)}_{\bd{j}_0 \dot\times \bd{k}}| = 0$, then the claim is fulfilled. Otherwise we know that $S(\bd{j}_0 \dot\times \bd{k}) = S$. Especially, this implies that $\bd{j}_0 \dot\times \bd{k} \notin \bd{K}(T)$ since $|T| > |S|$. By the definition of $\bd{K}(T)$, there is a set $\bar{\bd{J}} \subset \Jn(T)$ of $s^{|T|}$ indices such that for all $\bd{j} \in \bar{\bd{J}}$,
		\[
		|x_{\bd{j} \dot\times \bd{k}}| \geq |x_{\bd{j}_0 \dot\times \bd{k}}|.
		\]
		
		Assuming that $|x^{(S)}_{\bd{j}_0 \dot\times \bd{k}}|^2 > \frac{1}{s^{|T|}} \sum_{\bd{j} \in \Jn(T)} |x_{\bd{j} \dot\times \bd{k}}|^2$ implies
		
		\begin{align*}
			\sum_{\bd{j} \in \Jn(T)} |x_{\bd{j} \dot\times \bd{k}}|^2 \geq & \sum_{\bd{j} \in \bar{\bd{J}}} |x_{\bd{j} \dot\times \bd{k}}|^2 \geq
			\sum_{\bd{j} \in \bar{\bd{J}}} |x_{\bd{j}_0 \dot\times \bd{k}}|^2 \\
			= & |\bar{\bd{J}}| |x_{\bd{j}_0 \dot\times \bd{k}}|^2 \geq
			s^{|T|} |x^{(S)}_{\bd{j}_0 \dot\times \bd{k}}|^2  \\
			> & s^{|T|} \frac{1}{s^{|T|}} \sum_{\bd{j} \in \Jn(T)} |x_{\bd{j} \dot\times \bd{k}}|^2 =
			\sum_{\bd{j} \in \Jn(T)} |x_{\bd{j} \dot\times \bd{k}}|^2.
		\end{align*}
		
		This is a contradiction which completes the proof.
	\end{proof}

	\begin{lemma} \label{lem:max_sum_inequality_2}
		Let $S, T \subset [d]$ and $S \cap T = \emptyset$. Then for any index $\bd{k} \in \Jn([d] \backslash (S \cup T))$,
		\[
		\max_{\bd{j} \in \Jn(T)} \sum_{\bd{i} \in \Jn(S)} |x^{(S)}_{\bd{i} \dot\times \bd{j} \dot\times \bd{k}}|^2 \leq
		\frac{1}{s^{|T|}} \sum_{\bd{j} \in \Jn(T)} \sum_{\bd{i} \in \Jn(S)} |x_{\bd{i} \dot\times \bd{j} \dot\times \bd{k}}|^2.
		\]
	\end{lemma}
	
	\begin{proof}
		If $T = \emptyset$, then $s^{|T|} = 1$ and $\Jn(T)$ has exactly one element such that the claim holds since $|x^{(S)}_{\bd{i} \dot\times \bd{j} \dot\times \bd{k}}| \leq |x_{\bd{i} \dot\times \bd{j} \dot\times \bd{k}}|$ for any indices $\bd{i} \in \Jn(S)$, $\bd{j} \in \Jn(T)$, $\bd{k} \in \Jn([d] \backslash(S \cup T))$. So we can assume that $T \neq \emptyset$.
		
		Then we can apply Lemma \ref{lem:max_sum_inequality_1} to the sets $S$ and $S \cup T$. Since $S$ and $T$ are disjoint and $T \neq \emptyset$, $|S \cup T| > |S|$. The lemma yields that for any $\bd{k} \in \Jn([d] \backslash (S \cup T))$,
		\[
		\max_{\bd{j} \in \Jn(S \cup T)} |x^{(S)}_{\bd{j} \dot\times \bd{k}}|^2 \leq \frac{1}{s^{|S \cup T|}} \sum_{\bd{j} \in \Jn(S \cup T)} |x_{\bd{j} \dot\times \bd{k}}|^2.
		\]
		We can rewrite this as
		\begin{equation}
			\max_{\bd{i} \in \Jn(S)} \max_{\bd{j} \in \Jn(T)} |x^{(S)}_{\bd{i} \dot\times \bd{j} \dot\times \bd{k}}|^2 \leq
			\frac{1}{s^{|S| + |T|}} \sum_{\bd{i} \in \Jn(S)} \sum_{\bd{j} \in \Jn(T)} |x_{\bd{i} \dot\times \bd{j} \dot\times \bd{k}}|^2. \label{eq:comp_max_max_sum}
		\end{equation}
		
		By Lemma \ref{lem:non_0_entries}, for fixed $\bd{j} \in \Jn(T)$ and $\bd{k} \in \Jn([d] \backslash (S \cup T))$, there are at most $2^{|S|}$ indices $\bd{i} \in \Jn(S)$ such that $x^{(S)}_{\bd{i} \dot\times \bd{j} \dot\times \bd{k}} \neq 0$. Thus we obtain
		\[
		\max_{\bd{j} \in \Jn(T)} \sum_{\bd{i} \in \Jn(S)} |x^{(S)}_{\bd{i} \dot\times \bd{j} \dot\times \bd{k}}|^2 \leq 
		\max_{\bd{j} \in \Jn(T)} s^{|S|} \max_{\bd{i} \in \Jn(S)} |x^{(S)}_{\bd{i} \dot\times \bd{j} \dot\times \bd{k}}|^2
		\]
		Combining this with (\ref{eq:comp_max_max_sum}) yields
		\begin{align*}
			\max_{\bd{j} \in \Jn(T)} \sum_{\bd{i} \in \Jn(S)} |x^{(S)}_{\bd{i} \dot\times \bd{j} \dot\times \bd{k}}|^2 & \leq 
			\frac{s^{|S|}}{s^{|S| + |T|}} \sum_{\bd{i} \in \Jn(S)} \sum_{\bd{j} \in \Jn(T)} |x_{\bd{i} \dot\times \bd{j} \dot\times \bd{k}}|^2 \\
			& = 
			\frac{1}{s^{|T|}} \sum_{\bd{i} \in \Jn(S)} \sum_{\bd{j} \in \Jn(T)} |x_{\bd{i} \dot\times \bd{j} \dot\times \bd{k}}|^2.
		\end{align*}
	\end{proof}

	Let $N = n_1 n_2 \dots n_d$ and assume that the matrix $\Phi \in \R^{m \times N}$ has the $(2 s^d, \delta)$-RIP. We regard the matrix $\Phi^* \Phi - Id_N \in \R^{N \times N}$ as an array $\bd{B}$ of order $2 d$ with dimensions $\bd{n}^{\times 2} = (n_1, \dots, n_d, n_1, \dots, n_d)$ such that for all $\bd{i}, \bd{i}' \in \Jn$,
	
	\[
	B_{\bd{i} \dot+ \bd{i}'} = \sum_{k = 1}^m \Phi_{k, L^\bd{n}(\bd{i})} \Phi_{k, L^\bd{n}(\bd{i}')} - \mathbbm{1}_{\bd{i} = \bd{i}'}.
	\]
	
	Then for any arrays $\bd{x}, \bd{y} \in \R^\bd{n}$,
	\[
	\langle \Phi \vect(\bd{x}), \Phi \vect(\bd{y}) \rangle - \langle \vect(\bd{x}), \vect(\bd{y}) \rangle = \sum_{\bd{i}, \bd{i}' \in \Jn} B_{\bd{i} \dot+ \bd{i}'} x_{\bd{i}} y_{\bd{i}'}
	\]
	
	Let $\boldsymbol{\xi} \in \R^\bd{n}$ be the Rademacher tensor of order $d$, i.e., for $\bd{j} \in \Jn$,
	\[
	\xi_{\bd{j}} = \xi^{(1)}_{\bd{j}_1} \dots \xi^{(d)}_{\bd{j}_d}
	\]
	where $\xi^{(1)}, \dots, \xi^{(d)}$ are the independent Rademacher vectors from the assumption of Theorem \ref{thm:main}.
	Let $\bar{\xi}^{(1)}, \dots, \bar{\xi}^{(d)}$ and $\bar{\boldsymbol{\xi}}$ be corresponding independent copies.
	
	Consider the norm deviation represented by the chaos
	\[
	\tilde{X} := \langle A \vect(\bd{x}), A \vect(\bd{x}) \rangle - \langle \vect(\bd{x}), \vect(\bd{x}) \rangle = 
	\sum_{\bd{i}, \bd{i}' \in \Jn} B_{\bd{i} \dot+ \bd{i}'} x_{\bd{i}} x_{\bd{i}'} \xi_{\bd{i}} \xi_{\bd{i}'}
	\]
	
	and the corresponding decoupled chaos
	\[
	X := \sum_{\bd{i}, \bd{i}' \in \Jn} B_{\bd{i} \dot+ \bd{i}'} x_{\bd{i}} x_{\bd{i}'} \xi_{\bd{i}} \bar{\xi}_{\bd{i}'}
	\]
	
	Our goal is to bound the moments of $|X|$ which will also lead to bounds on the moments of $|\tilde{X} - \mathbb{E} \tilde{X}| = \left| \|A \vect(\bd{x})\|_2^2 - \|\bd{x}\|_2^2 - \mathbb{E}\left[ \|A \vect(\bd{x})\|_2^2 - \|\bd{x}\|_2^2 \right] \right|$ by the application of the decoupling Theorem \ref{thm:decoupling_2}. After showing that $\mathbb{E}\left[ \|A \vect(\bd{x})\|_2^2 - \|\bd{x}\|_2^2 \right]$ is sufficiently small, the moment bounds for $|X|$ will, in turn, lead to the proof of Theorem \ref{thm:main}.

	Fix $S, T \subset [d]$. Define the sums
	\begin{align*}
		\tilde{X}^{(S, T)} & := \langle A \vect(\bd{x}^{(S)}), A \vect(\bd{x}^{(T)}) \rangle - \langle \vect(\bd{x}^{(S)}), \vect(\bd{x}^{(T)}) \rangle \\
		& = 
		\sum_{\bd{i}, \bd{i}' \in \Jn} B_{\bd{i} \dot+ \bd{i}'} \xi_{\bd{i}} \xi_{\bd{i}'} x^{(S)}_{\bd{i}} x^{(T)}_{\bd{i}'}
	\end{align*}
	and their decoupled counterparts
	\[
	X^{(S, T)} := \sum_{\bd{i}, \bd{i}' \in \Jn} B_{\bd{i} \dot+ \bd{i}'} \xi_{\bd{i}} \bar{\xi}_{\bd{i}'} x^{(S)}_{\bd{i}} x^{(T)}_{\bd{i}'}
	\]
	such that
	\begin{align*}
		X = \sum_{S, T \subset [d]} X^{(S, T)} \quad \text{and} \quad \tilde{X} = \sum_{S, T \subset [d]} \tilde{X}^{(S, T)}.
	\end{align*}
	
	We obtain
	\begin{align*}
		X^{(S, T)} := \sum_{\substack{\bd{j} \in \Jn(S^c), \bd{k} \in \Jn(S) \\ \bd{j}' \in \Jn(T^c), \bd{k}' \in \Jn(T)}} B_{(\bd{j} \dot\times \bd{k}) \dot+ (\bd{j}' \dot\times \bd{k}')} \xi_{\bd{j} \dot\times \bd{k}} \bar{\xi}_{\bd{j}' \dot\times \bd{k}'} x^{(S)}_{\bd{j} \dot\times \bd{k}} x^{(T)}_{\bd{j}' \dot\times \bd{k}'}
	\end{align*}
	
	Now for any set $\tilde{I} \subset [d]$, define the array $\boldsymbol{\xi}^{(\tilde{I})} \in \R^{\bd{n}}(\tilde{I})$ by
	\[
	\xi^{(\tilde{I})}_{\bd{j}} = \prod_{l \in \tilde{I}} \xi^{(l)}_{\bd{j}_l}
	\]
	for any $\bd{j} \in \Jn(\tilde{I})$ and analogously for $\bar{\boldsymbol{\xi}}$.
	
	With this, we obtain
	\begin{equation}
		X^{(S, T)} = \sum_{\substack{\bd{j} \in \Jn(S^c), \bd{k} \in \Jn(S) \\ \bd{j}' \in \Jn(T^c), \bd{k}' \in \Jn(T)}} B_{(\bd{j} \dot\times \bd{k}) \dot+ (\bd{j}' \dot\times \bd{k}')} \xi^{(S)}_{\bd{k}} \xi^{(S^c)}_{\bd{j}} \bar{\xi}^{(T)}_{\bd{k}'} \bar{\xi}^{(T^c)}_{\bd{j}'} x^{(S)}_{\bd{j} \dot\times \bd{k}} x^{(T)}_{\bd{j}' \dot\times \bd{k}'}. \label{eq:ST_chaos}
	\end{equation}
	
	Note that $\boldsymbol{\xi}^{(S)},\, \boldsymbol{\xi}^{(S^c)},\, \bar{\boldsymbol{\xi}}^{(T)},\, \bar{\boldsymbol{\xi}}^{(T^c)}$ are independent. Condition on $\boldsymbol{\xi}^{(S)}$ and $\bar{\boldsymbol{\xi}}^{(T)}$ and treat (\ref{eq:ST_chaos}) as a chaos of order $2d - |S| - |T|$ with corresponding index array $\bd{B}^{(S, T)}$ given by
	\[
	B^{(S, T)}_{\bd{j} \dot\times \bd{j}'} = \sum_{\bd{k} \in \Jn(S), \bd{k}' \in \Jn(T)} B_{(\bd{j} \dot\times \bd{k}) \dot+ (\bd{j}' \dot\times \bd{k}')} \xi^{(S)}_{\bd{k}} \bar{\xi}^{(T)}_{\bd{k}'}  x^{(S)}_{\bd{j} \dot\times \bd{k}} x^{(T)}_{\bd{j}' \dot\times \bd{k}'}
	\]
	for $\bd{j} \in \Jn(S^c)$, $\bd{j}' \in \Jn(T^c)$. Thus, we need to control the norms appearing in the chaos concentration inequality for $\bd{B}^{(S, T)}$.

	\subsubsection{Application of Lemma \ref{lem:inner_pro_bound}}
	
	Let $1 \leq \kappa \leq 2d - |S| - |T|$ and $\bar{I}_1, \dots, \bar{I}_\kappa$ be a partition of the set $[2d] \backslash (S \cup (T + d) )$. Let $\bm{\alpha}^{(l)} \in \R^{\bd{n}^{\times 2}}(\bar{I}_l)$ and $\|\bm{\alpha}^{(l)}\|_2 = 1$ for $1 \leq l \leq \kappa$. By the definition of the norm,
	\[
	\|\bd{B}^{(S, T)}\|_{\bar{I}_1, \dots, \bar{I}_\kappa} = \sup_{\bm{\alpha}^{(1)}, \dots, \bm{\alpha}^{(\kappa)}} \sum_{\bd{i} \in \bd{J}^{\bd{n}^{\times 2}}(\bar{I}_1 \cup \dots \cup \bar{I}_\kappa)} B^{(S, T)}_{\bd{i}} \alpha^{(1)}_{\bd{i}_{\bar{I}_1}} \dots{} \alpha^{(\kappa)}_{\bd{i}_{\bar{I}_\kappa}}
	\]
	where the supremum is formed over all possible choices of the aforementioned arrays $\bm{\alpha}^{(1)}, \dots,  \bm{\alpha}^{(\kappa)}$.
	
	The partition sets $\bar{I}_1, \dots, \bar{I}_\kappa$ can contain elements of $[d]$ and of $[2d] \backslash [d]$. As outlined in Section \ref{sec:proof_outline}, we separate the partition sets by whether they intersect only $[d]$, only $[2 d] \backslash [d]$ or both of them. In this sense, we define $\Lambda := \{l \in [\kappa]: \bar{I}_l \subset [d]\}$, $\Lambda' := \{l \in [\kappa]: \bar{I}_l \subset [2 d] \backslash [d]\}$ and $\Gamma := [\kappa] \backslash (\Lambda \cup \Lambda')$. Then we join the corresponding partition sets to $\bar{I} := \bigcup_{l \in \Lambda} \bar{I}_l$, $\bar{I}' := \bigcup_{l \in \Lambda'} \bar{I}_l$ and $\bar{J} := \bigcup_{l \in \Gamma} \bar{I}_l$. Then the three sets $\bar{I},\,\bar{I}',\,\bar{J}$ also form a partition of the set $[2 d] \backslash (S \cup (T + d))$.
	
	Based on this, we define arrays $\bm{\alpha} \in \R^{\bd{n}^{\times 2}}(\bar{I}),\,\bm{\alpha}' \in \R^{\bd{n}^{\times 2}}(\bar{I}'),\,\bm{\beta} \in \R^{\bd{n}^{\times 2}}(\bar{J})$ by
	
	\begin{itemize}
		\item For each $\bd{i} \in \bd{J}^{\bd{n}^{\times 2}}(\bar{I})$,
		\[
		\alpha_{\bd{i}} = \prod_{l \in \Lambda} \alpha^{(l)}_{\bd{i}_{\bar{I}_l}}.
		\]
		
		\item For each $\bd{i} \in \bd{J}^{\bd{n}^{\times 2}}(\bar{I}')$,
		\[
		\alpha'_{\bd{i}} = \prod_{l \in \Lambda'} \alpha^{(l)}_{\bd{i}_{\bar{I}_l}}.
		\]
		
		\item For each $\bd{i} \in \bd{J}^{\bd{n}^{\times 2}}(\bar{J})$,
		\[
		\beta_{\bd{i}} = \prod_{l \in \Gamma} \alpha^{(l)}_{\bd{i}_{\bar{I}_l}}.
		\]
	\end{itemize}
	
	Then we have
	\begin{align*}
		\sum_{\bd{i} \in \bd{J}^{\times 2}(\bar{I}_1 \cup \dots \cup \bar{I}_\kappa)} B^{(S, T)}_{\bd{i}} \alpha^{(1)}_{\bd{i}_{\bar{I}_1}} \dots{} \alpha^{(\kappa)}_{\bd{i}_{\bar{I}_\kappa}} =
		\sum_{\bd{i} \in \bd{J}^{\bd{n}^{\times 2}}(\bar{I}), \bd{i}' \in \bd{J}^{\bd{n}^{\times 2}}(\bar{I}'), \bd{j} \in \bd{J}^{\bd{n}^{\times 2}}(\bar{J}) } B^{(S, T)}_{\bd{i} \dot\times \bd{i}' \dot\times \bd{j}} \alpha_{\bd{i}} \alpha'_{\bd{i}'} \beta_{\bd{j}}.
	\end{align*}
	and
	\begin{align*}
		\|\bm{\alpha}\|_2^2 & = \sum_{\bd{i} \in \bd{J}^{\bd{n}^{\times 2}}(\bar{I})} \left( \prod_{l \in \Lambda} \alpha^{(l)}_{\bd{i}_{\bar{I}_l}} \right)^2 = 
		\sum_{\bd{i} \in \bd{J}^{\bd{n}^{\times 2}}(\bar{I})} \prod_{l \in \Lambda} (\alpha^{(l)}_{\bd{i}_{\bar{I}_l}})^2 \\
		& = \prod_{l \in \Lambda} \sum_{\bd{i}^{(l)} \in \bd{J}^{\bd{n}^{\times 2}}(\bar{I}_l)}  (\alpha^{(l)}_{\bd{i}^{(l)}})^2 = 
		\prod_{l \in \Lambda} \|\bm{\alpha}^{(l)}\|_2^2 = 1.
	\end{align*}
	In the same way, it follows that $\|\bm{\alpha}'\|_2 = \|\bm{\beta}\|_2 = 1$.
	
	Thus, we can bound
	\begin{align*}
		\|\bd{B}^{(S, T)}\|_{\bar{I}_1, \dots, \bar{I}_\kappa} \leq \sup_{\bm{\alpha}, \bm{\alpha}', \bm{\beta}} \sum_{\bd{i} \in \bd{J}^{\bd{n}^{\times 2}}(\bar{I}), \bd{i}' \in \bd{J}^{\bd{n}^{\times 2}}(\bar{I}'), \bd{j} \in \bd{J}^{\bd{n}^{\times 2}}(\bar{J})} B^{(S, T)}_{\bd{i} \dot\times \bd{i}' \dot\times \bd{j}} \alpha_{\bd{i}} \alpha'_{\bd{i}'} \beta_{\bd{j}},
	\end{align*}
	where the supremum is formed over all $\bm{\alpha} \in \R^{\bd{n}^{\times 2}}(\bar{I})$, $\bm{\alpha}' \in \R^{\bd{n}^{\times 2}}(\bar{I}')$, $\bm{\beta} \in \R^{\bd{n}^{\times 2}}(\bar{J})$ with $\|\bm{\alpha}\|_2 = \|\bm{\alpha}'\|_2 = \|\bm{\beta}\|_2$.
	
	With the dual characterization of the $\ell_2$ norm, we can eliminate the supremum over $\bm{\beta}$, obtaining
	\begin{align}
		& \|\bd{B}^{(S, T)}\|_{\bar{I}_1, \dots, \bar{I}_\kappa} \leq \sup_{\bm{\alpha}, \bm{\alpha}'} \left[ \sum_{\bd{j} \in \bd{J}^{\bd{n}^{\times 2}}(\bar{J})} \left( \sum_{\bd{i} \in \bd{J}^{\bd{n}^{\times 2}}(\bar{I}), \bd{i}' \in \bd{J}^{\bd{n}^{\times 2}}(\bar{I}')} B^{(S, T)}_{\bd{i} \dot\times \bd{i}' \dot\times \bd{j}} \alpha_{\bd{i}} \alpha'_{\bd{i}'} \right)^2 \right]^\frac{1}{2}
		\label{eq:k_norm_1} \\
		= & \sup_{\bm{\alpha}, \bm{\alpha}'} \left[ \sum_{\substack{\bd{j} \in \bd{J}^{\bd{n}}(\bar{J} \cap [d]), \\ \bd{j}' \in \bd{J}^{\bd{n}}((\bar{J} - d) \cap [d])}} \left( \sum_{\substack{\bd{i} \in \bd{J}^{\bd{n}}(\bar{I}) \\ \bd{i}' \in \bd{J}^{\bd{n}}(\bar{I}' - d) \\ \bd{k} \in \bd{J}^{\bd{n}}(S) \\ \bd{k}' \in \bd{J}^{\bd{n}}(T)}} B_{\substack{(\bd{i} \dot\times \bd{j} \dot\times \bd{k}) \\ \dot+ (\bd{i}' \dot\times \bd{j}' \dot\times \bd{k}')}} \xi^{(S)}_{\bd{k}} \bar{\xi}^{(T)}_{\bd{k}'}  x^{(S)}_{\bd{i} \dot\times \bd{j} \dot\times \bd{k}} x^{(T)}_{\bd{i}' \dot\times \bd{j}' \dot\times \bd{k}'}  \alpha_{\bd{i}} \alpha'_{\bd{i}'} \right)^2 \right]^\frac{1}{2}
		\nonumber
	\end{align}
	
	Now for each $\bd{j} \in \bd{J}^{\bd{n}}(\bar{J} \cap [d])$ and $\bd{i} \in \Jn(\bar{I})$, define $\bd{u}^{(\bd{j}, \bd{i})}, \bar{\bd{x}}^{(\bd{j}, \bd{i})}, \bd{x}^{(\bd{j}, \bd{i})} \in \R^\bd{n}$ by
	\begin{align*}
		x^{(\bd{j}, \bd{i})}_{\bar{\bd{i}} \dot\times \bar{\bd{j}} \dot\times \bd{k}} = 
		\begin{cases}
			x_{\bd{i} \dot\times \bd{j} \dot\times \bd{k}} & \text{if } \bar{\bd{j}} = \bd{j} \text{ and } \bar{\bd{i}} = \bd{i} \\
			0 & \text{otherwise}
		\end{cases} \\
		\bar{x}^{(\bd{j}, \bd{i})}_{\bar{\bd{i}} \dot\times \bar{\bd{j}} \dot\times \bd{k}} = 
		\begin{cases}
			\xi^{(S)}_{\bd{k}} x^{(S)}_{\bd{i} \dot\times \bd{j} \dot\times \bd{k}} & \text{if } \bar{\bd{j}} = \bd{j} \text{ and } \bar{\bd{i}} = \bd{i} \\
			0 & \text{otherwise}
		\end{cases} \\
		u^{(\bd{j}, \bd{i})}_{\bar{\bd{i}} \dot\times \bar{\bd{j}} \dot\times \bd{k}} = 
		\bar{x}^{(\bd{j}, \bd{i})}_{\bar{\bd{i}} \dot\times \bar{\bd{j}} \dot\times \bd{k}} \alpha_{\bd{i}}
	\end{align*}
	for any $\bar{\bd{i}} \in \Jn(\bar{I})$, $\bar{\bd{j}} \in \Jn(\bar{J} \cap [d])$, $\bd{k} \in \Jn(S)$.

	In the same way, define for each $\bd{j}' \in \bd{J}^{\bd{n}}((\bar{J} - d) \cap [d])$ and $\bd{i}' \in \Jn(\bar{I}' - d)$, $\bd{v}^{(\bd{j}', \bd{i}')}, \bar{\bd{y}}^{(\bd{j}', \bd{i}')}, \bd{y}^{(\bd{j}', \bd{i}')} \in \R^\bd{n}$ by
	\begin{align*}
		y^{(\bd{j}', \bd{i}')}_{\bar{\bd{i}}' \dot\times \bar{\bd{j}}' \dot\times \bd{k}'} = 
		\begin{cases}
			x_{\bd{i}' \dot\times \bd{j}' \dot\times \bd{k}'} & \text{if } \bar{\bd{j}}' = \bd{j}' \text{ and } \bar{\bd{i}}' = \bd{i}' \\
			0 & \text{otherwise}
		\end{cases} \\
		\bar{y}^{(\bd{j}', \bd{i}')}_{\bar{\bd{i}}' \dot\times \bar{\bd{j}}' \dot\times \bd{k}'} = 
		\begin{cases}
			\xi^{(T)}_{\bd{k}'} x^{(T)}_{\bd{i}' \dot\times \bd{j}' \dot\times \bd{k}'} & \text{if } \bar{\bd{j}}' = \bd{j}' \text{ and } \bar{\bd{i}}' = \bd{i}' \\
			0 & \text{otherwise}
		\end{cases} \\
		v^{(\bd{j}', \bd{i}')}_{\bar{\bd{i}}' \dot\times \bar{\bd{j}}' \dot\times \bd{k}'} = 
		\bar{y}^{(\bd{j}', \bd{i}')}_{\bar{\bd{i}}' \dot\times \bar{\bd{j}}' \dot\times \bd{k}'} \alpha'_{\bd{i}'}
	\end{align*}
	for any $\bar{\bd{i}}' \in \Jn(\bar{I}')$, $\bar{\bd{j}}' \in \Jn((\bar{J} - d) \cap [d])$, $\bd{k}' \in \Jn(T)$.
	
	Then for $\bd{j} \in \bd{J}^{\bd{n}}(\bar{J} \cap [d])$, $\bd{i} \in \Jn(\bar{I})$, $\bd{j}' \in \bd{J}^{\bd{n}}((\bar{J} - d) \cap [d])$, $\bd{i}' \in \Jn(\bar{I}' - d)$,
	\begin{align*}
		& \sum_{\bd{l} \in \Jn, \bd{l}' \in \Jn} B_{\bd{l} \dot+ \bd{l}} u^{(\bd{j}, \bd{i})}_{\bd{l}} v^{(\bd{j}', \bd{i}')}_{\bd{l}'} \\
		= & \sum_{\bd{k} \in \Jn(S), \bd{k}' \in \Jn(T)} B_{(\bd{i} \dot\times \bd{j} \dot\times \bd{k}) \dot+ (\bd{i}' \dot\times \bd{j}' \dot\times \bd{k}')} \xi^{(S)}_{\bd{k}} \bar{\xi}^{(T)}_{\bd{k}'}  x^{(S)}_{\bd{j} \dot\times \bd{k}} x^{(T)}_{\bd{j}' \dot\times \bd{k}'}  \alpha_{\bd{i}} \alpha'_{\bd{i}'}
	\end{align*}
	
	Substituting into (\ref{eq:k_norm_1}), we obtain
	\begin{align}
		\|\bd{B}^{(S, T)}\|_{\bar{I}_1, \dots, \bar{I}_\kappa} \leq
		\sup_{\bm{\alpha}, \bm{\alpha}'} \left[ \sum_{\substack{\bd{j} \in \bd{J}^{\bd{n}}(\bar{J} \cap [d]) \\ \bd{j}' \in \bd{J}^{\bd{n}}((\bar{J} - d) \cap [d])}} \left( \sum_{\substack{\bd{i} \in \Jn(\bar{I}) \\ \bd{i}' \in \Jn(\bar{I}' - d)}} \sum_{\bd{l} \in \Jn, \bd{l}' \in \Jn} B_{\bd{l} \dot+ \bd{l}} u^{(\bd{j}, \bd{i})}_{\bd{l}} v^{(\bd{j}', \bd{i}')}_{\bd{l}'} \right)^2 \right]^\frac{1}{2}
		\label{eq:k_norm_2}
	\end{align}
	
	We will define a vectorized version of $\bd{u}^{(\bd{j}, \bd{i})}$ and $\bd{v}^{(\bd{j}', \bd{i}')}$ and then apply Lemma \ref{lem:inner_pro_bound}.
	
	Now define $\bar{n}_1 := \prod_{l \in \bar{J} \cap [d]} n_l$, $\bar{n}_2 := \prod_{l \in \bar{I}} n_l$, $\bar{n}_3 := \prod_{l \in S} n_l$ and analogously $\bar{n}'_1 := \prod_{l \in (\bar{J} - d) \cap [d]} n_l$, $\bar{n}'_2 := \prod_{l \in (\bar{I}' - d)} n_l$, $\bar{n}'_3 := \prod_{l \in T} n_l$. Then $\bar{n}_1 \bar{n}_2 \bar{n}_3 = \bar{n}'_1 \bar{n}'_2 \bar{n}'_3 = N$.
	
	
	Now we rearrange the previously defined arrays as vectors. Note that functions $\hat{L}^{\bd{n}}_{I}$ for $I \subset [d]$ map integers to partial array indices (tuples).
	
	For $j \in [\bar{n}_1],\, k \in [\bar{n}_2]$, define the vectors
	\begin{align*}
		x^{(j, k)} := \vect(\bd{x}^{(\hat{L}^{\bd{n}}_{\bar{J} \cap [d]}(j), \hat{L}^{\bd{n}}_{\bar{I}}(k))}) \in \R^N \\
		\bar{x}^{(j, k)} := \vect(\bar{\bd{x}}^{(\hat{L}^{\bd{n}}_{\bar{J} \cap [d]}(j), \hat{L}^{\bd{n}}_{\bar{I}}(k))}) \in \R^N \\
		u^{(j, k)} := \vect(\bd{u}^{(\hat{L}^{\bd{n}}_{\bar{J} \cap [d]}(j), \hat{L}^{\bd{n}}_{\bar{I}}(k))}) \in \R^N
	\end{align*}
	and for $j' \in [\bar{n}'_1],\, k' \in [\bar{n}'_2]$,
	\begin{align*}
		y^{(j', k')} := \vect(\bd{y}^{(\hat{L}^{\bd{n}}_{(\bar{J} - d) \cap [d]}(j'), \hat{L}^{\bd{n}}_{\bar{I}' - d}(k'))}) \in \R^N \\
		\bar{y}^{(j', k')} := \vect(\bar{\bd{v}}^{(\hat{L}^{\bd{n}}_{(\bar{J} - d) \cap [d]}(j'), \hat{L}^{\bd{n}}_{\bar{I}' - d}(k'))}) \in \R^N \\
		v^{(j', k')} := \vect(\bd{v}^{(\hat{L}^{\bd{n}}_{(\bar{J} - d) \cap [d]}(j'), \hat{L}^{\bd{n}}_{\bar{I}' - d}(k'))}) \in \R^N.
	\end{align*}
	
	Using this, we can write (\ref{eq:k_norm_2}) as
	\begin{equation}
		\|\bd{B}^{(S, T)}\|_{\bar{I}_1, \dots, \bar{I}_\kappa} \leq
		\sup_{\bm{\alpha}, \bm{\alpha}'} \left[ \sum_{j \in [\bar{n}_1], j' \in [\bar{n}'_1]} \left( \sum_{k \in [\bar{n}_2], k' \in [\bar{n}'_2]} (u^{(j, k)})^* B v^{(j', k')} \right)^2 \right]^\frac{1}{2}. \label{eq:k_norm_3}
	\end{equation}
	where $B \in \R^{N \times N}$ is the array $\bd{B}$ rearranged as a matrix, i.e., $B = \Phi^* \Phi - Id_N$.

	To check the other requirements of Lemma \ref{lem:inner_pro_bound}, define $s_1 := s^{|\bar{J} \cap [d]|}$, $s_2 := s^{|\bar{I}|}$, $s_3 := s^{|S|}$, $s_1' := s^{|(\bar{J} - d) \cap [d]|}$, $s_2' := s^{|\bar{I}' - d|}$, $s_3' := s^{|T|}$. We obtain $s_1 s_2 s_3 = s_1' s_2' s_3' = s^d$ and by assumption $\Phi$ satisfies the $(2 s^d, \delta)$-RIP.

	It follows from their definition that within each of the three families $(\bd{u}^{(\bd{j}, \bd{i})})$, $(\bar{\bd{x}}^{(\bd{j}, \bd{i})})$, $(\bd{x}^{(\bd{j}, \bd{i})})$, the arrays have disjoint supports and the same also holds for the vectorized versions $(u^{(j, k)})$, $(\bar{x}^{(j, k)})$, $(x^{(j, k)})$. In an analogous way, also $(v^{(j', k')})$, $(\bar{y}^{(j', k')})$, $(y^{(j', k')})$ have disjoint supports.

	Fix $(j, k) \in [\bar{n}_1] \times [\bar{n}_2]$ and the corresponding $\bd{j} = \hat{L}^{\bd{n}}_{\bar{J} \cap [d]}(j)$, $\bd{i} = \hat{L}^{\bd{n}}_{\bar{I}}(k)$. Then the number of non-$0$ entries of $\bar{x}^{(j, k)}$ is the number of $\bd{k} \in \Jn(S)$ such that $x^{(S)}_{\bd{i} \dot\times \bd{j} \dot\times \bd{k}} \neq 0$. By Lemma \ref{lem:non_0_entries}, this is $\leq s^{|S|} = s_3$. Thus, $\bar{x}^{(j,k)}$ and $u^{(j, k)}$ are $s_3$-sparse. In the same way, it also holds that $\bar{y}^{(j', k')}$ and $v^{(j', k')}$ are $s_3'$-sparse for all $(j', k') \in [\bar{n}_1'] \times [\bar{n}_2']$.

	Now define $\bd{b}_2(1) \subset \Jn(\bar{I})$ as the set of the $s^{|\bar{I}|}$ indices $\bd{i} \in \Jn(\bar{I})$ for which
	$\sum_{\bd{j} \in \Jn(\bar{J} \cap [d])} \|\bd{x}^{(\bd{j}, \bd{i})}\|_2$ attains the largest values, $\bd{b}_2(2) \subset \Jn(\bar{I})$ as the set of the next $s^{|\bar{I}|}$ remaining indices with the largest corresponding values and so on until $\bd{b}_2(R_2)$ where $R_2 = \lceil \frac{\bar{n}_2}{s_2} \rceil$. In the same way, define $\bd{b}_2'(1), \dots, \bd{b}_2'(R_2')$ associated to $\bd{y}^{(\bd{j}, \bd{i})}$ instead of $\bd{x}^{(\bd{j}, \bd{i})}$.
	
	Furthermore, for $K \in [R_2]$, define the sets $\bd{b}_1^{(K)}(1), \dots, \bd{b}_1^{(K)}(R_1) \subset \Jn(\bar{J} \cap [d])$ (blocks of size $s^{|\bar{J} \cap [d]|}$) by sorting the indices $\bd{j} \in  \Jn(\bar{J} \cap [d])$ according to $\sum_{\bd{i} \in \bd{b}_2(K)} \|\bd{u}^{(\bd{j}, \bd{i})}\|_2^2$. In the same way, $\bd{b}_1^{(K')'}(1), \dots, \bd{b}_1^{(K')'}(R_1') \subset \Jn((\bar{J} - d) \cap [d])$ is defined for $K' \in [R_2']$ based on $\sum_{\bd{i}' \in \bd{b}'_2(K')} \|\bd{v}^{(\bd{j}', \bd{i}')}\|_2^2$ for $\bd{j}' \in \Jn((\bar{J} - d) \cap [d])$.
	
	With the aforementioned rearrangement functions, these partitions of array indices give rise to corresponding partitions 
	of $[\bar{n}_1]$, $[\bar{n}_1']$, $[\bar{n}_2]$, $[\bar{n}_2']$
	as required in the prerequisite of Lemma \ref{lem:inner_pro_bound}.

	With this, we have shown all assumptions of Lemma \ref{lem:inner_pro_bound}. So we can apply it to (\ref{eq:k_norm_3}) and obtain that there is a suitable choice of $\bd{\alpha}, \bd{\alpha}'$ such that

	\begin{align}
		& \|\bd{B}^{(S, T)}\|_{\bar{I}_1, \dots, \bar{I}_\kappa} \nonumber\\
		\leq & \delta \left[ \sum_{(J, K, \bar{K}) \in [R_1] \times [R_2]^2} \sqrt{  \sum_{j \in b^{(K)}_{1}(J)} \|u^{(j, (K))}\|_2^2 \|u^{(j, \bar{K})}\|_2^2} \right]^\frac{1}{2} \nonumber\\
		& \cdot \left[ \sum_{(J', K', \bar{K}') \in [R_1'] \times [R_2']^2} \sqrt{ \sum_{j' \in b^{(K')'}_{1'}(J')} \|v^{(j', (K'))}\|_2^2 \|v^{(j', (\bar{K}'))}\|_2^2 } \right]^\frac{1}{2}
		\label{eq:result_of_lemma}
	\end{align}
	with the notations as in Lemma \ref{lem:inner_pro_bound}.

	In terms of the previously defined arrays and index partitions, we can write the first factor in brackets in (\ref{eq:result_of_lemma}) as
	\[
	\sum_{(J, K, \bar{K}) \in [R_1] \times [R_2]^2} \sqrt{ \sum_{\bd{j} \in \bd{b}_1^{(K)}(J)} \left( \sum_{\bd{i} \in \bd{b}_2(K)} \|\bd{u}^{(\bd{j}, \bd{i})}\|_2^2 \right) \left( \sum_{\bd{i} \in \bd{b}_2(\bar{K})} \|\bd{u}^{(\bd{j}, \bd{i})}\|_2^2 \right) }.
	\]

	\subsubsection{Bounding the norms using the block structure}
	
	\begin{lemma} \label{lem:lem_next_bound}
		With the notation of the previous paragraphs, it holds that 
		\begin{align*}
			\sum_{(J, K, \bar{K}) \in [R_1] \times [R_2]^2} \sqrt{ \sum_{\bd{j} \in \bd{b}_1^{(K)}(J)} \left( \sum_{\bd{i} \in \bd{b}_2(K)} \|\bd{u}^{(\bd{j}, \bd{i})}\|_2^2 \right) \left( \sum_{\bd{i} \in \bd{b}_2(\bar{K})} \|\bd{u}^{(\bd{j}, \bd{i})}\|_2^2 \right) } \leq \frac{4 \sqrt{2}}{\sqrt{s_1} s_2}.
		\end{align*}
	\end{lemma}
	
	\begin{proof}
		First, we observe the following things
		\begin{itemize}
			\item For every $\bd{j} \in \Jn(\bar{J} \cap [d])$, it holds that
			\begin{align}
				& \max_{\bd{i} \in \Jn(\bar{I})} \|\bar{\bd{x}}^{(\bd{j}, \bd{i})}\|_2^2 =
				\max_{\bd{i} \in \Jn(\bar{I})} \sum_{\bd{k} \in \Jn(S)} |\xi_{\bd{k}}^{(S)} x^{(S)}_{\bd{i} \dot\times \bd{j} \dot\times \bd{k}} |^2 =
				\max_{\bd{i} \in \Jn(\bar{I})} \sum_{\bd{k} \in \Jn(S)} |x^{(S)}_{\bd{i} \dot\times \bd{j} \dot\times \bd{k}} |^2 \nonumber \\
				\leq & \frac{1}{s^{|\bar{I}|}} \sum_{\bd{i} \in \Jn(\bar{I})} \sum_{\bd{k} \in \Jn(S)} |x_{\bd{i} \dot\times \bd{j} \dot\times \bd{k}} |^2 = \frac{1}{s_2} \sum_{\bd{i} \in \Jn(\bar{I})} \|\bd{x}^{(\bd{j}, \bd{i})}\|_2^2,
				\label{eq:max_sum_xbar_j}
			\end{align}
			where we used Lemma \ref{lem:max_sum_inequality_2} in the third step.
			
			\item Analogously, for every $\bd{i} \in \Jn(\bar{I})$,
			\begin{align}
				& \max_{\bd{j} \in \Jn(\bar{J} \cap [d])} \|\bar{\bd{x}}^{(\bd{j}, \bd{i})}\|_2^2 = 
				\max_{\bd{j} \in \Jn(\bar{J} \cap [d])} \sum_{\bd{k} \in \Jn(S)} |x^{(S)}_{\bd{i} \dot\times \bd{j} \dot\times \bd{k}} |^2 \nonumber \\
				\leq & \frac{1}{s^{|\bar{J} \cap [d]|}} \sum_{\bd{j} \in \Jn(\bar{J} \cap [d])} \sum_{\bd{k} \in \Jn(S)} |x^{(S)}_{\bd{i} \dot\times \bd{j} \dot\times \bd{k}} |^2 =
				\frac{1}{s_1} \sum_{\bd{j} \in \Jn(\bar{J} \cap [d])} \|\bd{x}^{(\bd{j}, \bd{i})}\|_2^2.
				\label{eq:max_sum_xbar_i}
			\end{align}
			
			\item Finally, Lemma \ref{lem:max_sum_inequality_2} also implies
			\begin{align}
				& \max_{\substack{\bd{j} \in \Jn(\bar{J} \cap [d]) \\ \bd{i} \in \Jn(\bar{I})}} \|\bar{\bd{x}}^{(\bd{j}, \bd{i})}\|_2^2 = 
				\max_{\bd{i} \in \Jn(\bar{I})} \max_{\bd{j} \in \Jn(\bar{J} \cap [d])} \sum_{\bd{k} \in \Jn(S)} |x^{(S)}_{\bd{i} \dot\times \bd{j} \dot\times \bd{k}} |^2 \nonumber \\
				\leq & \frac{1}{s^{|\bar{I} \cup (\bar{J} \cap [d])|}} \sum_{\bd{i} \in \Jn(\bar{I})} \sum_{\bd{j} \in \Jn(\bar{J} \cap [d])} \sum_{\bd{k} \in \Jn(S)} |x^{(S)}_{\bd{i} \dot\times \bd{j} \dot\times \bd{k}} |^2 =
				\frac{1}{s_1 s_2} \sum_{\substack{\bd{j} \in \Jn(\bar{J} \cap [d]) \\ \bd{i} \in \Jn(\bar{I})}} \|\bd{x}^{(\bd{j}, \bd{i})}\|_2^2. 
				\label{eq:max_sum_xbar_ij}
			\end{align}
		\end{itemize}

		Now applying Holder's inequality to the sum inside the square root on the left hand side of Lemma \ref{lem:lem_next_bound} gives
		\begin{align*}
			& \sum_{(J, K, \bar{K}) \in [R_1] \times [R_2]^2} \sqrt{ \sum_{\bd{j} \in \bd{b}_1^{(K)}(J)} \left( \sum_{\bd{i} \in \bd{b}_2(K)} \|\bd{u}^{(\bd{j}, \bd{i})}\|_2^2 \right) \left( \sum_{\bd{i} \in \bd{b}_2(\bar{K})} \|\bd{u}^{(\bd{j}, \bd{i})}\|_2^2 \right) } \\
			\leq & \sum_{(K, \bar{K}) \in [R_2]^2} \sum_{J \in [R_1]} \sqrt{ \max_{\bd{j} \in \bd{b}_1^{(K)}(J)} \left( \sum_{\bd{i} \in \bd{b}_2(K)} \|\bd{u}^{(\bd{j}, \bd{i})}\|_2^2 \right)} \cdot \sqrt{ \sum_{\bd{j} \in \bd{b}_1^{(K)}(J)} \left( \sum_{\bd{i} \in \bd{b}_2(\bar{K})} \|\bd{u}^{(\bd{j}, \bd{i})}\|_2^2 \right) }.
		\end{align*}
		
		Then we can apply the Cauchy-Schwarz inequality to the sum over $J \in [R_1]$ and bound this by
		\begin{align*}
			& \sum_{(K, \bar{K}) \in [R_2]^2} \sqrt{ \sum_{J \in [R_1]} \max_{\bd{j} \in \bd{b}_1^{(K)}(J)} \left( \sum_{\bd{i} \in \bd{b}_2(K)} \|\bd{u}^{(\bd{j}, \bd{i})}\|_2^2 \right)} \cdot \sqrt{ \sum_{\substack{J \in [R_1] \\ \bd{j} \in \bd{b}_1^{(K)}(J)}} \left( \sum_{\bd{i} \in \bd{b}_2(\bar{K})} \|\bd{u}^{(\bd{j}, \bd{i})}\|_2^2 \right) } \\
			& = \sum_{(K, \bar{K}) \in [R_2]^2} \sqrt{ \left( \sum_{J \in [R_1]} \max_{\bd{j} \in \bd{b}_1^{(K)}(J)}  \sum_{\bd{i} \in \bd{b}_2(K)} \|\bd{u}^{(\bd{j}, \bd{i})}\|_2^2 \right) \cdot  \left( \sum_{\bd{i} \in \bd{b}_2(\bar{K})} \sum_{\bd{j} \in \Jn(\bar{J} \cap [d])} \|\bd{u}^{(\bd{j}, \bd{i})}\|_2^2 \right) }.
		\end{align*}
		
		The definition of $\bd{b}_1^{(K)}(J)$ implies that for every $J \in [R_1] \backslash \{1\}$ and every $\bd{j} \in \bd{b}_1^{(K)}(J)$, $\bar{\bd{j}} \in \bd{b}_1^{(K)}(J - 1)$, it holds that $\sum_{\bd{i} \in \bd{b}_2(K)} \|\bd{u}^{(\bd{j}, \bd{i})}\|_2^2 \leq \sum_{\bd{i} \in \bd{b}_2(K)} \|\bd{u}^{(\bar{\bd{j}}, \bd{i})}\|_2^2$. Together with $|\bd{b}_1^{(K)}(J - 1)| = s_1$, we obtain for $J \geq 2$, $\max_{\bd{j} \in \bd{b}_1^{(K)}(J)} \sum_{\bd{i} \in \bd{b}_2(K)} \|\bd{u}^{(\bd{j}, \bd{i})}\|_2^2 \leq \frac{1}{s_1} \sum_{\bd{j} \in  \bd{b}_1^{(K)}(J - 1)} \sum_{\bd{i} \in \bd{b}_2(K)} \|\bd{u}^{(\bd{j}, \bd{i})}\|_2^2$. Using this together with separately considering the $J = 1$ term yields the bound
		\begin{align}
			& \sum_{K, \bar{K} \in [R_2]} \sqrt{\left[ \sum_{\substack{J \in [R_1] \backslash \{1\} \\ \bd{j} \in \bd{b}_1^{(K)}(J - 1) \\ \bd{i} \in \bd{b}_2(K)}} \frac{\|\bd{u}^{(\bd{j}, \bd{i})}\|_2^2}{s_1} + \max_{\bd{j} \in \Jn(\bar{J} \cap [d])} \sum_{\bd{i} \in \bd{b}_2(K)} \|\bd{u}^{(\bd{j}, \bd{i})}\|_2^2 \right] \cdot \sum_{\substack{\bd{i} \in \bd{b}_2(\bar{K}) \\ \bd{j} \in \Jn(\bar{J} \cap [d])}} \|\bd{u}^{(\bd{j}, \bd{i})}\|_2^2 }
			\nonumber \\
			& \leq \sum_{K, \bar{K} \in [R_2]} \sqrt{ \left[ \sum_{\substack{\bd{i} \in \bd{b}_2(K) \\ \bd{j} \in \Jn(\bar{J} \cap [d])}}  \frac{\|\bd{u}^{(\bd{j}, \bd{i})}\|_2^2}{s_1} + \max_{\bd{j} \in \Jn(\bar{J} \cap [d])} \sum_{\bd{i} \in \bd{b}_2(K)} \|\bd{u}^{(\bd{j}, \bd{i})}\|_2^2 \right] \cdot \sum_{\substack{\bd{i} \in \bd{b}_2(K) \\ \bd{j} \in \Jn(\bar{J} \cap [d])}} \|\bd{u}^{(\bd{j}, \bd{i})}\|_2^2 } \nonumber \\
			& =: \sum_{K, \bar{K} \in [R_2]} \sqrt{ W_{K, \bar{K}} } \label{eq:boundu}.
		\end{align}
		
		We obtain
		\begin{align*}
			& \max_{\bd{j} \in \Jn(\bar{J} \cap [d])} \left( \sum_{\bd{i} \in \bd{b}_2(K)}   \|\bd{u}^{(\bd{j}, \bd{i})}\|_2^2 \right) \leq
			\sum_{\bd{i} \in \bd{b}_2(K)} \max_{\bd{j} \in \Jn(\bar{J} \cap [d])}  \|\bd{u}^{(\bd{j}, \bd{i})}\|_2^2 \\
			\leq & \sum_{\bd{i} \in \bd{b}_2(K)} |\alpha_{\bd{i}}|^2 \max_{\bd{j} \in \Jn(\bar{J} \cap [d])}  \|\bar{\bd{x}}^{(\bd{j}, \bd{i})}\|_2^2 \leq
			\|\bm{\alpha}_{\bd{b}_2(K)}\|_2^2 \cdot \max_{\bd{i} \in \bd{b}_2(K)} \max_{\bd{j} \in \Jn(\bar{J} \cap [d])}  \|\bar{\bd{x}}^{(\bd{j}, \bd{i})}\|_2^2. 
		\end{align*}
		
		If $K = 1$, we can bound this using (\ref{eq:max_sum_xbar_ij}),
		\[
		\|\bm{\alpha}_{\bd{b}_2(K)}\|_2^2 \cdot \max_{\bd{i} \in \Jn(\bar{I})} \max_{\bd{j} \in \Jn(\bar{J} \cap [d])}  \|\bar{\bd{x}}^{(\bd{j}, \bd{i})}\|_2^2 \leq \frac{1}{s_1 s_2} \|\bm{\alpha}_{\bd{b}_2(K)}\|_2^2 \cdot \sum_{(\bd{j}, \bd{i}) \in \Jn(\bar{J} \cap [d]) \times \Jn(\bar{I})} \|\bd{x}^{(\bd{j}, \bd{i})}\|_2^2.
		\]
		
		For $K \geq 2$, we can bound it using (\ref{eq:max_sum_xbar_i}),
		\begin{align*}
			\|\bm{\alpha}_{\bd{b}_2(K)}\|_2^2 \cdot \max_{\bd{i} \in \bd{b}_2(K)} \max_{\bd{j} \in \Jn(\bar{J} \cap [d])}  \|\bar{\bd{x}}^{(\bd{j}, \bd{i})}\|_2^2 \leq
			\frac{1}{s_1} \|\bm{\alpha}_{\bd{b}_2(K)}\|_2^2 \cdot \max_{\bd{i} \in \bd{b}_2(K)} \sum_{\bd{j} \in \Jn(\bar{J} \cap [d])} {\|\bd{x}^{(\bd{j}, \bd{i})}\|_2^2}.
		\end{align*}
		
		Also, for any $K \geq 1$,
		
		\begin{align*}
			& \sum_{\bd{i} \in \bd{b}_2(K)} \sum_{\bd{j} \in \Jn(\bar{J} \cap [d])} \|\bd{u}^{(\bd{j}, \bd{i})}\|_2^2 \leq 
			\sum_{\bd{i} \in \bd{b}_2(K)} |\alpha_{\bd{i}}|^2 \sum_{\bd{j} \in \Jn(\bar{J} \cap [d])} \|\bar{\bd{x}}^{(\bd{j}, \bd{i})}\|_2^2 \\
			& \leq \|\bm{\alpha}_{\bd{b}_2(K)}\|_2^2 \cdot \max_{\bd{i} \in \bd{b}_2(K)} \sum_{\bd{j} \in \Jn(\bar{J} \cap [d])} \|\bar{\bd{x}}^{(\bd{j}, \bd{i})}\|_2^2.
		\end{align*}
		
		Note that we always have $\|\bar{\bd{x}}^{(\bd{j}, \bd{i})}\|_2 \leq \|\bd{x}^{(\bd{j}, \bd{i})}\|_2$. So for the terms $W_{K, \bar{K}}$ in (\ref{eq:boundu}), this implies for all $K, \bar{K} \geq 2$,
		\begin{align*}
			W_{K, \bar{K}} \leq & \frac{2}{s_1} \left[ \|\bm{\alpha}_{\bd{b}_2(K)}\|_2^2 \cdot \max_{\bd{i} \in \bd{b}_2(K)} \sum_{\bd{j} \in \Jn(\bar{J} \cap [d])} \|\bd{x}^{(\bd{j}, \bd{i})}\|_2^2 \right] \\
			& \cdot \left[ \|\bm{\alpha}_{\bd{b}_2(\bar{K})}\|_2^2 \cdot \max_{\bd{i} \in \bd{b}_2(\bar{K})} \sum_{\bd{j} \in \Jn(\bar{J} \cap [d])} \|\bar{\bd{x}}^{(\bd{j}, \bd{i})}\|_2^2 \right] \\
			W_{K, 1} \leq & \frac{2}{s_1} \left[ \|\bm{\alpha}_{\bd{b}_2(K)}\|_2^2 \cdot \max_{\bd{i} \in \bd{b}_2(K)} \sum_{\bd{j} \in \Jn(\bar{J} \cap [d])} \|\bd{x}^{(\bd{j}, \bd{i})}\|_2^2 \right] \\
			& \cdot \left[ \|\bm{\alpha}_{\bd{b}_2(1)}\|_2^2 \cdot \max_{\bd{i} \in \Jn(\bar{I})} \sum_{\bd{j} \in \Jn(\bar{J} \cap [d])} \|\bar{\bd{x}}^{(\bd{j}, \bd{i})}\|_2^2 \right] \\
			W_{1, \bar{K}} \leq & \frac{\|\bm{\alpha}_{\bd{b}_2(1)}\|_2^2}{s_1} \left[ \max_{\bd{i} \in \Jn(\bar{I})} \sum_{\bd{j} \in \Jn(\bar{J} \cap [d])} \|\bar{\bd{x}}^{(\bd{j}, \bd{i})}\|_2^2 + \frac{1}{s_2} \sum_{\substack{\bd{j} \in \Jn(\bar{J} \cap [d]) \\ \bd{i} \in \Jn(\bar{I})}} \|\bd{x}^{(\bd{j}, \bd{i})}\|_2^2  \right] \\
			& \cdot \left[ \|\bm{\alpha}_{\bd{b}_2(\bar{K})}\|_2^2 \cdot \max_{\bd{i} \in \bd{b}_2(\bar{K})} \sum_{\bd{j} \in \Jn(\bar{J} \cap [d])} \|\bar{\bd{x}}^{(\bd{j}, \bd{i})}\|_2^2 \right] \\
			W_{1, 1} \leq & \frac{\|\bm{\alpha}_{\bd{b}_2(1)}\|_2^2}{s_1} \left[ \max_{\bd{i} \in \Jn(\bar{I})} \sum_{\bd{j} \in \Jn(\bar{J} \cap [d])} \|\bar{\bd{x}}^{(\bd{j}, \bd{i})}\|_2^2 + \frac{1}{s_2} \sum_{\substack{\bd{j} \in \Jn(\bar{J} \cap [d]) \\ \bd{i} \in \Jn(\bar{I})}} \|\bd{x}^{(\bd{j}, \bd{i})}\|_2^2  \right] \\
			& \cdot \left[ \|\bm{\alpha}_{\bd{b}_2(1)}\|_2^2 \cdot \max_{\bd{i} \in \Jn(\bar{I})} \sum_{\bd{j} \in \Jn(\bar{J} \cap [d])} \|\bar{\bd{x}}^{(\bd{j}, \bd{i})}\|_2^2 \right].
		\end{align*}
		
		Then we bound (\ref{eq:boundu}) by
		\begin{align*}
			& \sum_{K, \bar{K} \in [R_2] \backslash \{1\}} \sqrt{W_{K, \bar{K}}} + 
			\sum_{K \in [R_2] \backslash \{1\}} \sqrt{W_{K, 1}} + 
			\sum_{\bar{K} \in [R_2] \backslash \{1\}} \sqrt{W_{1, \bar{K}}} + 
			\sqrt{W_{1, 1}} \\
			& =: \textcircled{1} + \textcircled{2} + \textcircled{3} + \textcircled{4}.
		\end{align*}
		
		For part $\textcircled{1}$, we obtain,
		\begin{align*}
			\textcircled{1} = \sqrt\frac{2}{s_1} \left[  \sum_{K \in [R_2] \backslash \{1\})} \sqrt{ \|\bm{\alpha}_{\bd{b}_2(K)}\|_2^2 \cdot \max_{\bd{i} \in \bd{b}_2(K)} \sum_{\bd{j} \in \Jn(\bar{J} \cap [d])} \|\bd{x}^{(\bd{j}, \bd{i})}\|_2^2 } \right]^2
		\end{align*}
		where we can apply Holder's inequality on the sum over $K$, giving
		\begin{align*}
			\textcircled{1} & \leq \sqrt\frac{2}{s_1} \left[ \sum_{K \in [R_2] \backslash \{1\}} \|\bm{\alpha}_{\bd{b}_2(K)}\|_2^2 \right] \left[ \sum_{K \in [R_2] \backslash \{1\}} \max_{\bd{i} \in \bd{b}_2(K)} \sum_{\bd{j} \in \Jn(\bar{J} \cap [d])} \|\bd{x}^{(\bd{j}, \bd{i})}\|_2^2 \right] \\
			& \leq \sqrt\frac{2}{s_1} \left[ \sum_{K \in [R_2] \backslash \{1\}} \max_{\bd{i} \in \bd{b}_2(K)} \sum_{\bd{j} \in \Jn(\bar{J} \cap [d])} \|\bd{x}^{(\bd{j}, \bd{i})}\|_2^2 \right]
		\end{align*}
		
		The definition of $\bd{b}_2(K)$ yields that for every $K \geq 2$,
		\[
		\max_{\bd{i} \in \bd{b}_2(K)} \sum_{\bd{j} \in \Jn(\bar{J} \cap [d])} \|\bd{x}^{(\bd{j}, \bd{i})}\|_2^2 \leq \frac{1}{s_2} \sum_{\bd{i} \in \bd{b}_2(K - 1)} \sum_{\bd{j} \in \Jn(\bar{J} \cap [d])} \|\bd{x}^{(\bd{j}, \bd{i})}\|_2^2.
		\]
		Then
		\begin{align*}
			\textcircled{1} \leq \sqrt\frac{2}{s_1} \frac{1}{s_2} \left[ \sum_{K \in [R_2] \backslash \{1\}} \sum_{\bd{i} \in \bd{b}_2(K - 1)} \sum_{\bd{j} \in \Jn(\bar{J} \cap [d])} \|\bd{x}^{(\bd{j}, \bd{i})}\|_2^2 \right] \leq \sqrt\frac{2}{s_1} \frac{1}{s_2} \|\bd{x}\|_2^2.
		\end{align*}
		
		Next note that by (\ref{eq:max_sum_xbar_j}), 
		\begin{align*}
			\max_{\bd{i} \in \Jn(\bar{I})} \sum_{\bd{j} \in \Jn(\bar{J} \cap [d])} \|\bar{\bd{x}}^{(\bd{j}, \bd{i})}\|_2^2 & \leq
			\sum_{\bd{j} \in \Jn(\bar{J} \cap [d])} \max_{\bd{i} \in \Jn(\bar{I})}  \|\bar{\bd{x}}^{(\bd{j}, \bd{i})}\|_2^2 \\
			& \leq 
			\frac{1}{s_2} \sum_{\bd{j} \in \Jn(\bar{J} \cap [d]), \bd{i} \in \Jn(\bar{I})}  \|\bd{x}^{(\bd{j}, \bd{i})}\|_2^2 = \frac{1}{s_2} \|\bd{x}\|_2^2.
		\end{align*}
		So for the second term we obtain using the Cauchy-Schwarz inequality,
		\begin{align*}
			\textcircled{2} & \leq \sqrt\frac{2}{s_1 s_2} \|\bm{\alpha}_{\bd{b}_2(1)}\|_2 \|\bd{x}\|_2 \sum_{K \in [R_2] \backslash \{1\}} \sqrt{ \|\bm{\alpha}_{\bd{b}_2(K)}\|_2^2 \cdot \max_{\bd{i} \in \bd{b}_2(K)} \sum_{\bd{j} \in \Jn(\bar{J} \cap [d])} \|\bd{x}^{(\bd{j}, \bd{i})}\|_2^2 } \\
			\leq & \sqrt\frac{2}{s_1 s_2} \|\bm{\alpha}_{\bd{b}_2(1)}\|_2 \|\bd{x}\|_2  \sqrt{ \sum_{K \in [R_2] \backslash \{1\}} \|\bm{\alpha}_{\bd{b}_2(K)}\|_2^2 \cdot \sum_{K \in [R_2] \backslash \{1\}} \max_{\bd{i} \in \bd{b}_2(K)} \sum_{\bd{j} \in \Jn(\bar{J} \cap [d])} \|\bd{x}^{(\bd{j}, \bd{i})}\|_2^2 } \\
			\leq & \sqrt\frac{2}{s_1 s_2} \|\bd{x}\|_2  \sqrt{ \frac{1}{s_2} \sum_{K \in [R_2] \backslash \{1\}} \sum_{\bd{i} \in \bd{b}_2(K - 1)} \sum_{\bd{j} \in \Jn(\bar{J} \cap [d])} \|\bd{x}^{(\bd{j}, \bd{i})}\|_2^2 } \\
			\leq & \sqrt\frac{2}{s_1} \frac{1}{s_2} \|\bd{x}\|_2^2,
		\end{align*}
		
		Similarly, we can bound the other parts,	
		\begin{align*}
			& \textcircled{3} \\
			& \leq \frac{\|\bm{\alpha}_{\bd{b}_2(1)}\|_2}{\sqrt{s_1}} \sqrt{\frac{1}{s_2} \|\bd{x}\|_2^2 + \frac{1}{s_2} \|\bd{x}\|_2^2} \sum_{\bar{K} \in [R_2] \backslash \{1\}} \sqrt{\|\bm{\alpha}_{\bd{b}_2(\bar{K})}\|_2^2 \cdot \max_{\bd{i} \in \bd{b}_2(\bar{K})} \sum_{\bd{j} \in \Jn(\bar{J} \cap [d])} \|\bd{x}^{(\bd{j}, \bd{i})}\|_2^2} \\
			& \leq \sqrt\frac{2}{s_1 s_2} \|\bd{x}\|_2 \sqrt{\sum_{\bar{K} \in [R_2] \backslash \{1\}} \|\bm{\alpha}_{\bd{b}_2(\bar{K})}\|_2^2 \cdot \sum_{\bar{K} \in [R_2] \backslash \{1\}}  \max_{\bd{i} \in \bd{b}_2(\bar{K})} \sum_{\bd{j} \in \Jn(\bar{J} \cap [d])} \|\bd{x}^{(\bd{j}, \bd{i})}\|_2^2} \\
			& \leq \sqrt\frac{2}{s_1 s_2} \|\bd{x}\|_2 \sqrt{\frac{1}{s_2} \sum_{\bar{K} \in [R_2] \backslash \{1\}}  \sum_{\bd{i} \in \bd{b}_2(\bar{K} - 1)} \sum_{\bd{j} \in \Jn(\bar{J} \cap [d])} \|\bd{x}^{(\bd{j}, \bd{i})}\|_2^2} \\
			& \leq \sqrt\frac{2}{s_1} \frac{1}{s_2} \|\bd{x}\|_2^2,
		\end{align*}

		\begin{align*}
			\textcircled{4} \leq \frac{\|\bm{\alpha}_{\bd{b}_2(1)}\|_2}{\sqrt{s_1}} \sqrt{\frac{1}{s_2} \|\bd{x}\|_2^2 + \frac{1}{s_2} \|\bd{x}\|_2^2} \|\bm{\alpha}_{\bd{b}_2(1)}\|_2 \sqrt{\frac{1}{s_2} \|\bd{x}\|_2^2} \leq \sqrt\frac{2}{s_1} \frac{1}{s_2} \|\bd{x}\|_2^2.
		\end{align*}
		
		Altogether, it follows that (\ref{eq:boundu}) can be bounded by
		\[
		\frac{4 \sqrt{2}}{\sqrt{s_1} s_2} \|\bd{x}\|_2^2
		\]
		where $\|\bd{x}\|_2^2 = 1$. This completes the proof of the lemma.

	\end{proof}
	
	Now in (\ref{eq:result_of_lemma}), both factors in brackets can be bounded using Lemma \ref{lem:lem_next_bound} and we obtain
	\begin{align}
		\|\bd{B}^{(S, T)}\|_{\bar{I}_1, \dots, \bar{I}_\kappa} \leq \delta \frac{4 \sqrt{2}}{(s_1 s_1')^\frac{1}{4} (s_2 s_2')^\frac{1}{2}}. \label{eq:k_norm_4}
	\end{align}

	Note that
	\[
	(s_1 s_1')^\frac{1}{4} (s_2 s_2')^\frac{1}{2} =
	s^{\frac{1}{4}(|\bar{J} \cap [d]| + |(\bar{J} - d) \cap [d]|) + \frac{1}{2} (|\bar{I}| + |\bar{I}'|)} =
	s^{\frac{1}{4}|\bar{J}| + \frac{1}{2} (|\bar{I}| + |\bar{I}'|)}.
	\]
	
	Our goal is to show $s^{\frac{1}{4}|\bar{J}| + \frac{1}{2} (|\bar{I}| + |\bar{I}'|)} \geq s^\frac{\kappa}{2}$. This is given by the following lemma:
	
	\begin{lemma} \label{lem:partition_counting}
		Let $S, T \subset [d]$ and $\bar{I}_1, \dots, \bar{I}_\kappa$ be a partition of $[2 d] \backslash (S \cup (T + d))$ into $\kappa$ non-empty sets and $\bar{I} = \bigcup_{j \in [\kappa]: \bar{I}_j \subset [d]} \bar{I}_j$,  $\bar{I}' = \bigcup_{j \in [\kappa]: \bar{I}_j \subset ([2 d] \backslash [d])} \bar{I}_j$, $\bar{J} = [2 d] \backslash (S \cup (T + d) \cup \bar{I} \cup \bar{I}')$. Then
		\[
		\frac{1}{4} |\bar{J}| + \frac{1}{2}(|\bar{I}| + |\bar{I}'|) \geq \frac{\kappa}{2}.
		\]
	\end{lemma}
	
	\begin{proof}
		Note that $|\bar{J}| + |\bar{I}| + |\bar{I}'| = 2 d - |S| - |T|$. If $\kappa \leq d - \frac{|S| + |T|}{2}$,
		\[
		\frac{1}{4}|\bar{J}| + \frac{1}{2} (|\bar{I}| + |\bar{I}'|) \geq \frac{1}{4} (|\bar{J}| + |\bar{I}| + |\bar{I}'|) = \frac{2 d - |S| - |T|}{4} = \frac{d - \frac{|S|}{2} - \frac{|T|}{2} }{2} \geq \frac{\kappa}{2}.
		\]
		
		Now assume that $\kappa > d - \frac{|S| + |T|}{2}$. Let $\kappa' \leq \kappa$ be the number of indices $l \in [\kappa]$ such that $|\bar{I}_l| = 1$. All other sets $\bar{I}_l$ must contain at least two elements and the total number of elements is $2 d - |S| - |T| = \sum_{l \in [\kappa]} |\bar{I}_l| \geq \kappa' + 2(\kappa - \kappa') = 2 \kappa - \kappa'$. This implies that $\kappa' \geq 2 (\kappa - d) + |S| + |T|$. Every one-element set $\bar{I}_l$ is completely contained in either $[d]$ or $[2d] \backslash [d]$ and thus $\bar{I}_l \subset \bar{I}$ or $\bar{I}_l \subset \bar{I}'$. In this case, $|\bar{I}| + |\bar{I}'| \geq \kappa' \geq 2(\kappa - d) + |S| + |T|$ and we obtain
		\begin{align*}
			& \frac{1}{4}|\bar{J}| + \frac{1}{2} (|\bar{I}| + |\bar{I}'|) =
			\frac{1}{4}(|\bar{J}| + |\bar{I}| + |\bar{I}'|) + \frac{1}{4} (|\bar{I}| + |\bar{I}'|) \\
			& \geq \frac{1}{4} \cdot (2 d - |S| - |T|) + \frac{1}{4} \cdot (2 (\kappa - d) + |S| + |T|) = \frac{\kappa}{2}
		\end{align*}
		which completes the proof.
	\end{proof}
	
	Now $\frac{1}{4}|\bar{J}| + \frac{1}{2} (|\bar{I}| + |\bar{I}'|) \geq \frac{\kappa}{2}$ implies that $(s_1 s_1')^\frac{1}{4} (s_2 s_2')^\frac{1}{2} \geq s^\frac{\kappa}{2}$ and substituting into (\ref{eq:k_norm_4}) results in
	\[
	\|\bd{B}^{(S, T)}\|_{\bar{I}_1, \dots, \bar{I}_\kappa} \leq \frac{4 \sqrt{2} \delta}{s^\frac{\kappa}{2}} \|\bd{x}\|_2^2.
	\]
	
	With the moment bound (Theorem \ref{thm:subgaussian_decoupled}), this implies that for $p \geq 2$,
	\[
	\|X^{(S, T)}\|_{L_p} \leq C_1(d) \sum_{\kappa = 1}^{2 d} p^\frac{\kappa}{2} \frac{\delta}{s^\frac{\kappa}{2}}
	\]
	with $C_1(d)$ depending only on $d$ and we obtain
	\[
	\|X\|_{L_p} \leq \sum_{S, T \subset [d]} \|X^{(S, T)}\|_{L_p} \leq C_2(d) \sum_{\kappa = 1}^{2 d} p^\frac{\kappa}{2} \frac{\delta}{s^\frac{\kappa}{2}}.
	\]
	
	This moment bound holds for all $\|\bd{x}\|_2 = 1$ so by Theorem \ref{thm:decoupling_2}, we also obtain
	\[
	\|\tilde{X} - \mathbb{E} \tilde{X} \|_{L_p} \leq C_3(d) \sum_{\kappa = 1}^{2 d} p^\frac{\kappa}{2} \frac{\delta}{s^\frac{\kappa}{2}}.
	\]
	
	\subsubsection{Completing the proof}
	
	To bound the expectation, note that
	\begin{align*}
		|\mathbb{E} \tilde{X} |
		& = \left| \mathbb{E} \left[
		\sum_{\bd{i}, \bd{i}' \in \Jn} B_{\bd{i} \dot+ \bd{i}'} x_{\bd{i}} x_{\bd{i}'} \prod_{l \in [d]} \xi^{(l)}_{\bd{i}_l} \xi^{(l)}_{\bd{i}'_l} \right] \right|
		= \left| \sum_{\bd{i}, \bd{i}' \in \Jn} B_{\bd{i} \dot+ \bd{i}'} x_{\bd{i}} x_{\bd{i}'} \prod_{l \in [d]} \mathbb{E}[\xi^{(l)}_{\bd{i}_l} \xi^{(l)}_{\bd{i}'_l}] \right| \\
		& = \left| \sum_{\bd{i}, \bd{i}' \in \Jn} B_{\bd{i} \dot+ \bd{i}'} x_{\bd{i}} x_{\bd{i}'} \prod_{l \in [d]} \mathbbm{1}_{\bd{i}_l = \bd{i}'_l} \right|
		= \left| \sum_{\bd{i} \in \Jn} B_{\bd{i} \dot+ \bd{i}} x_{\bd{i}}^2 \right|
		\leq \max_{\bd{i} \in \Jn} |B_{\bd{i} \dot+ \bd{i}}| \cdot \|\bd{x}\|_2^2 \\
		& = \max_{j \in [N]} |(\Phi^* \Phi - Id_N)_{j j}|.
	\end{align*}
	
	Applying the RIP of $\Phi$ to the $1$-sparse canonical basis vectors $e_1, \dots, e_N$, this implies
	\begin{align*}
		|\mathbb{E} \tilde{X} | \leq \max_{j \in [N]} |e_j^* (\Phi^* \Phi - Id_N) e_j|
		= \max_{j \in [N]} \left| \|\Phi e_j\|_2^2 - \|e_j\|_2^2 \right| \leq \delta,
	\end{align*}
	such that $|\tilde{X}| \leq |\tilde{X} - \mathbb{E}\tilde{X}| + \delta$.
	Then using Lemma \ref{lem:moment_tail_bound} with the bound on $\|\tilde{X} - \mathbb{E} \tilde{X} \|_{L_p}$, we obtain for $\epsilon \geq \delta$,
	\begin{align*}
		\mathbb{P}(|\tilde{X}| > \epsilon) & \leq e^2 \exp\left(-\min_{\kappa \in [2d]} \left( \frac{(\epsilon - \delta) s^\frac{\kappa}{2} }{C_4(d) \delta} \right)^\frac{2}{\kappa} \right) \\
		& = e^2 \exp\left(- s \min_{\kappa \in [2d]} \left( \frac{\epsilon - \delta}{C_4(d) \delta} \right)^\frac{2}{\kappa} \right)
	\end{align*}
	
	So for choosing $\epsilon := \left[ C_4(d) + 1 \right] \delta$, we obtain
	\[
	\mathbb{P}\left( \left| \|A \vect(\bd{x})\|_2^2 - \|\bd{x}\|_2^2 \right| > \epsilon \right) \leq e^2 e^{-s}.
	\]
	
	By choosing $s \geq \log\frac{1}{\eta} + 2$, this is $\leq \eta$ which completes the proof.

	\section{Lower bounds} \label{sec:lower_bounds}
	
	The goal of this section is to show that our results, especially Corollary \ref{cor:jl_hadamard} that we obtain for Hadamard matrices, are optimal with respect to the probability $\eta$. To do this, we apply the Tensor randomized subsampled Hadamard transform to a set of $p$ points. By a union bound and Corollary \ref{cor:jl_hadamard}, this randomized transform simultaneously preserves the norms of $p$ vectors simultaneously with probability $1 - \nu$ if
	\[
	m \geq C(d) \epsilon^{-2} \left(\log \frac{p}{\nu}\right)^d \left( \log \frac{C(d)}{\epsilon} \right)^2 \log N \left(\log \frac{C(d) \log\frac{p}{\nu}}{\epsilon} \right)^2
	\]
	and
	\[
	N \geq \frac{1}{\nu^{C_1 d \log \log \frac{p}{\nu}}}.
	\]
	We will prove that the dependence $m \gtrsim (\log p)^d$ on $p$ (neglecting double logarithmic factors) is optimal.
	
	Regard the Hadamard transform $H = \mathbb{R}^{N \times N}$ as the Fourier transform on $\mathbb{F}_2^n$ where $N = 2^n$.
	That is,
	\[
	H_{j k} = (H_n)_{j k} = \frac{1}{\sqrt{N}} (-1)^{\langle j - 1, k - 1 \rangle_b},
	\]
	where $\langle a, b \rangle_b$ denotes the inner product of the binary representations of $a$ and $b$.
	
	Our approach is based on the special behavior of the Hadamard matrix on indicator vectors of subspaces of $\mathbb{F}_2^n$. This principle has been used before to show lower bounds for the restricted isometry property of subsampled Hadamard matrices in \cite{blasiok2019improved} and was then adapted to Johnson-Lindenstrauss embeddings in \cite{bamberger2021optimal}.
	
	Denote $\mathbb{G}_{n, r}$ for the set of all $r$-dimensional subspaces of $\mathbb{F}_2^n$. For any subset $M \subset \mathbb{F}_2^n$ we write $\mathbbm{1}_M \in \mathbb{R}^{N}$ for the indicator vector of $M$ normalized such that $\| \mathbbm{1}_M \|_2 = 1$. With this notation, it holds for any $V \in \mathbb{G}_{n, r}$ that (see Lemma II.1 in \cite{blasiok2019improved})
	\[
	H \mathbbm{1}_V = \mathbbm{1}_{V^\perp}
	\]

	Let $P_\Omega  \in \mathbb{R}^{m \times N}$ be the matrix representing subsampling $m$ out of $N$ entries independently and uniformly with replacement and rescaling by $\sqrt\frac{N}{m}$.

	Let $N = N_1 \cdot \dots \cdot N_d$, $N_j = 2^{n_j}$ for $1 \leq j \leq d$. Consider the matrix
	\[
	A = P_\Omega F D_\xi
	\]
	where $\xi = \xi^{(1)} \otimes \dots \otimes \xi^{(d)}$ is a Kronecker product of $d$ Rademacher vectors, $\xi^{(j)} \in \{\pm 1\}^{N_j}$ and $F \in \mathbb{R}^{N \times N}$ is a bounded orthonormal matrix.
	
	Let $2 \leq r \leq \min \{n_1, \dots, n_d \}$ and $s = 2^r$. For each $1 \leq j \leq d$ consider a subspace $V_j \in \mathbb{G}_{n_j, r}$. By taking $F := H_{n_1} \otimes \dots \otimes H_{n_d}$ and $x := \mathbbm{1}_{V_1} \otimes \dots \otimes \mathbbm{1}_{V_d}$, we obtain
	\[
	y = F x = (H_{n_1} \mathbbm{1}_{V_1}) \otimes \dots \otimes (H_{n_d} \mathbbm{1}_{V_d}) =
	\mathbbm{1}_{V_1^\perp} \otimes \dots \otimes \mathbbm{1}_{V_d^\perp}.
	\]
	
	The vector $y$ has $\frac{N_1}{s} \cdot \dots \cdot \frac{N_d}{s} = \frac{N}{s^d}$ entries of size $\sqrt\frac{s^d}{N}$. In subsampling with replacement, each selected entry is $\sqrt\frac{s^d}{N}$ with probability $\frac{1}{s^d}$ and $0$ with probability $1 - \frac{1}{s^d}$. Then
	\[
	\mathbb{P}(P_\Omega y = 0) = (1 - \frac{1}{s^d})^m \geq \exp(- \frac{2 m}{s^d}).
	\]
	
	Now consider the set $E := \{(D_{\hat{\xi}^{(1)}} \mathbbm{1}_{V_1}) \otimes \dots \otimes (D_{\hat{\xi}^{(d)}} \mathbbm{1}_{V_d}) | \hat{\xi}^{(1)} \in \{\pm 1\}^{N_1}, \dots, \hat{\xi}^{(d)} \in \{\pm 1\}^{N_d}  \}$. Corresponding to each factor, there are $2^s$ sign patterns such that $p := |E| \leq 2^{d s}$.
	
	With respect to the matrix $A = P_\Omega F D_\xi$, we note that for any value of the random vector $\xi$, there exists $\hat{x} \in E$ such that $D_\xi \hat{x} = x$. Then $A \hat{x} = P_\Omega F D_\xi \hat{x} = P_\Omega y$. We obtain $A \hat{x} = 0$ with probability $\geq$
	\begin{equation}
		\exp\left(- \frac{2 m}{s^d} \right) = \exp\left(- 2 m \left(\frac{d \log 2}{\log p} \right)^d \right) \label{prob_bound}
	\end{equation}
	with respect to the randomness in $P_\Omega$.
	
	Altogether, with the probability (\ref{prob_bound}),
	\[
	\sup_{x \in E} \left| \|A x\|_2 - 1 \right| \geq 1,
	\]
	i.e., the Johnson-Lindenstrauss condition is violated.
	
	To achieve that (\ref{prob_bound}) is $\leq \nu$, we need that $m \geq \frac{1}{2} (\log\frac{1}{\nu}) (\frac{\log p}{d \log 2})^d$.

	\section{Conclusions and implications for oblivious sketching} \label{sec:discussion}
	
	Our approach provides a sharp generalization of the near equivalence  between Johnson-Lindenstrauss property and restricted isometry property from \cite{riptojl}; the special case $d = 1$ in our work recovers the result of \cite{riptojl}.  We prove the Johnson-Lindenstrauss property without any assumption on the vectors it is applied to, i.e., it is not necessary for them to have Kronecker structure.   As Section \ref{sec:lower_bounds} shows, Corollary \ref{cor:jl_hadamard} is optimal with respect to the dependence on the probability $\eta$ even for vectors with Kronecker structure, implying that even for this case, the dependence on the required sparsity level $s$ on $\eta$ in Theorem \ref{thm:main} is optimal.
	
	With this provably optimal $\eta$ dependence, Corollary \ref{cor:jl_hadamard} also provides an improvement compared to Lemma 4.11 in \cite{ahle2020oblivious}. In that work, the construction $P_\Omega H D_\xi$ as in Corollary \ref{cor:jl_hadamard} is introduced as TensorSRHT and is used as one element of a more extensive fast embedding for vectors with Kronecker structure which allows for a computational complexity that is only polynomial in the degree $d$. This embedding is based on a tree structure. Starting from a vector $x = x^{(1)} \otimes \dots \otimes x^{(d)}$ with Kronecker structure, first a sparse Johnson-Lindenstrauss transform (OSNAP) is applied to each $x^{(j)}$ from $n$ to $m \geq m_1 = \Theta(\epsilon^{-2} \log\frac{1}{\eta})$ dimensions (Lemma 4.8 in \cite{ahle2020oblivious}). Subsequently the TensorSRHT is applied to $\frac{d}{2}$ pairs of these vectors, reducing the corresponding Kronecker products of two factors separately. In this way, the result is a reduced Kronecker product of $\frac{d}{2}$ factors. This reduction is applied successively until only a single factor remains at the end. In each level, the TensorSRHT acts as an embedding $\R^{m^2} \rightarrow \R^m$ for a suitable $m \geq m_1$. As such, the dimension is reduced from $m^d$ to $m$ after the application of OSNAP.
	
	Observe that this construction uses the setup of Corollary \ref{cor:jl_hadamard} for the case $d = 2$ and $N = m^2$.
	Choose 
	\[
	m := \left\lceil C \epsilon^{-2} \left( \log \frac{1}{\epsilon} \right)^2 \left( \log \frac{1}{\eta} \right)^2 \left( \log \frac{\log\frac{1}{\eta}}{\epsilon}\right)^3 \right\rceil.
	\]
	Then for sufficiently large $C$, $m \geq m_1$ such that OSNAP provides a suitable embedding $\R^{m^2} \rightarrow \R^{m}$. Also, as required by the aforementioned construction, after choosing the RIP matrix with constant success probability, Corollary \ref{cor:jl_hadamard} provides an embedding $\R^{m^2} \rightarrow \R^{m}$ satisfying the $(\epsilon, \eta)$-distributional Johnson-Lindenstrauss property since the required embedding dimension is
	{
		\begin{align*}
			m'  &= C' \epsilon^{-2} \left( \log\frac{1}{\epsilon} \right)^2 \left(\log\frac{1}{\eta}\right)^2 \left(\log (m^2) \right) \left( \log\frac{\log\frac{1}{\eta}}{\epsilon} \right)^2 \\
			& \leq 2 \tilde{C}' (\log C) \epsilon^{-2} \left( \log\frac{1}{\epsilon} \right)^2 \left(\log\frac{1}{\eta}\right)^2 \left( \log\frac{\log\frac{1}{\eta}}{\epsilon} \right)^3
		\end{align*}
	}
	which is $\leq m$ for sufficiently large $C$. So omitting $\log\frac{1}{\epsilon}$ and $\log\log\frac{1}{\eta}$ factors, our result requires an embedding dimension $m$ of $\Omega\left(\epsilon^{-2} \left(\log\frac{1}{\eta} \right)^2 \right)$ compared to the dimension $\Omega\left(\epsilon^{-2} \left(\log\frac{1}{\eta} \right)^3 \right)$ in \cite{ahle2020oblivious}. Thus, our result leads to both an improved embedding power and, consequently, an improved computational complexity of the tensor computation procedure. 

	
	\section*{Acknowledgments}
	R.W. is supported by AFOSR MURI FA9550-19-1-0005, NSF DMS 1952735, and NSF IFML 2019844. S.B. and F.K. have been supported by the German Science Foundation (DFG) in the context of the Emmy-Noether Junior Research Group KR4512/1-2.
	R.W. and F.K. gratefully acknowledge support from the Institute for Advanced Study, where this project was initiated.


	\printbibliography

\end{document}